\documentclass[journal,onecolumn]{IEEEtran}
\usepackage{amsthm,amsmath,amssymb,graphicx,enumerate,amsfonts,array,url,cite}
\usepackage{refcount,bm,multirow,bigstrut,dblfloatfix}
\newtheorem{theorem}{Theorem}[section]
\newtheorem{lemma}{Lemma}[section]

\newtheorem{definition}{Definition}[section]
\newtheorem{remark}{Remark}[section]

\newcommand{\specialcell}[2][c]{\begin{tabular}[#1]{@{}c@{}}#2\end{tabular}}

\def \X{{\bm{X}}}
\def \x{{\bm{x}}}
\def \Y{{\bm{Y}}}
\def \y{{\bm{y}}}
\def \Z{{\bm{Z}}}

\def \Sone{{S^1}}
\def \Stwo{{S^2}}
\def \bS{{\beta_S}}
\def \bone{b_1}
\def \btwo{b_2}
\def \bmin{{b_{\min}}}
\def \bmax{{b_{\max}}}

\def \SNR{{\text{SNR}}}

\def \So{{S_\omega}}

\def \St{{\tilde{S}}}
\def \SSt{{S \setminus \tilde{S}}}

\def \dx{\, \mathrm{d}}

\def \om{\omega}
\def \omh{\hat{\omega}}

\allowdisplaybreaks

\begin{document}

\title{Sparse Signal Processing with Linear and Nonlinear Observations: A Unified Shannon-Theoretic Approach}

\author{Cem Aksoylar, George Atia and Venkatesh Saligrama%
  \thanks{This work was partially supported by NSF Grants CCF-1320547 and CNS-1330008, the U.S.\ Department of Homeland Security, Science and Technology Directorate, Office of University Programs, under Grant Award 2013-ST-061-ED0001, by ONR Grant 50202168 and US AF contract FA8650-14-C-1728. The material in this paper was presented in part at the 2013 IEEE Information Theory Workshop, Seville, Spain, September 2013.}}

\IEEEpubid{Copyright~\copyright~2016 IEEE}

\maketitle

\begin{abstract}
  We derive fundamental sample complexity bounds for recovering sparse and structured signals for linear and nonlinear observation models including sparse regression, group testing, multivariate regression and problems with missing features. In general, sparse signal processing problems can be characterized in terms of the following Markovian property. We are given a set of $N$ variables $X_1,X_2,\ldots,X_N$, and there is an unknown subset of variables $S \subset \{1,\ldots,N\}$ that are \emph{relevant} for predicting outcomes $Y$. More specifically, when $Y$ is conditioned on $\{X_n\}_{n\in S}$ it is conditionally independent of the other variables, $\{X_n\}_{n \not \in S}$. 
  Our goal is to identify the set $S$ from samples of the variables $X$ and the associated outcomes $Y$. We characterize this problem as a version of the noisy channel coding problem. Using asymptotic information theoretic analyses, we establish mutual information formulas that provide sufficient and necessary conditions on the number of samples required to successfully recover the salient variables. These mutual information expressions unify conditions for both linear and nonlinear observations. We then compute sample complexity bounds for the aforementioned models, based on the mutual information expressions in order to demonstrate the applicability and flexibility of our results in general sparse signal processing models.
\end{abstract}

\section{Introduction}

In this paper, we are concerned with the asymptotic analysis of the sample complexity in problems where we aim to identify a set of \emph{salient} variables responsible for producing an outcome. In particular, we assume that among a set of $N$ variables/features $X = (X_1,\ldots,X_N)$, only $K$ variables (indexed by set $S$) are directly relevant to the outcome $Y$. We formulate this concept in terms of Markovianity, namely, given $X_S = \{X_n\}_{n \in S}$, the outcome $Y$ is independent of the other variables $\{X_n\}_{n \not\in S}$, i.e.,
\begin{equation} \label{eq:markov}
P(Y|X,S) = P(Y|X_S,S).
\end{equation}
We also explicitly consider the existence of an independent latent random quantity affecting the observation model, which we denote with $\bS$. Similar to \eqref{eq:markov}, with this latent factor we have the observation model
\begin{equation}\label{eq:markov_b}
  P(Y | X, \bS, S) = P(Y | X_S, \bS, S).
\end{equation}
Note that the existence of such a latent factor does not violate \eqref{eq:markov}.
Abstractly, the set of salient variables $S$ is generated from a distribution over a collection of sets, ${\cal S}$, of size $K$. For unstructured sparse problems, ${\cal S}$ is the collection of all $K$-sets in $\{1,\ldots,N\}$ while for structured problems ${\cal S}$ is a subset of this unstructured $K$-set collection. The latent factor $\bS$ (if it exists) is generated from a distribution $p(\bS) \triangleq P(\bS \mid S)$. Then independent and identically distributed (IID) samples of $X$ are generated from distribution $Q(X)$, and for each sample an observation $Y$ is generated using the conditional distribution $P(Y | X_S, \bS, S)$ conditioned on $X$, $\bS$ and $S$, as in \eqref{eq:markov_b}. The set $S$ and $\bS$ are fixed across different samples of variables $X$ and corresponding observations $Y$.

We assume we are given $T$ sample pairs $(X, Y)$ denoted as $(\X, \Y)$ and the problem is to identify the set of salient variables, $S$, from these $T$ samples given the knowledge of the observation model $P(Y|X_S,\bS,S)$ and $p(\bS)$. Our analysis aims to establish necessary and sufficient conditions on $T$ in order to recover the set $S$ with an arbitrarily small error probability in terms of $K$, $N$, the observation model and other model parameters such as the signal-to-noise ratio. We consider the average error probability, where the average is over the random $S$, $\bS$, $\X$ and $\Y$. In this paper, we limit our analysis to the setting with IID variables $X$ for simplicity. It turns out that our methods can be extended to the dependent case at the cost of additional terms in our derived formulas that compensate for dependencies between the variables. Some results derived for the former setting were presented in \cite{ssp} and more recently in \cite{aistats}.

\begin{figure}[h]
  \centering
  \includegraphics[width=0.4\textwidth]{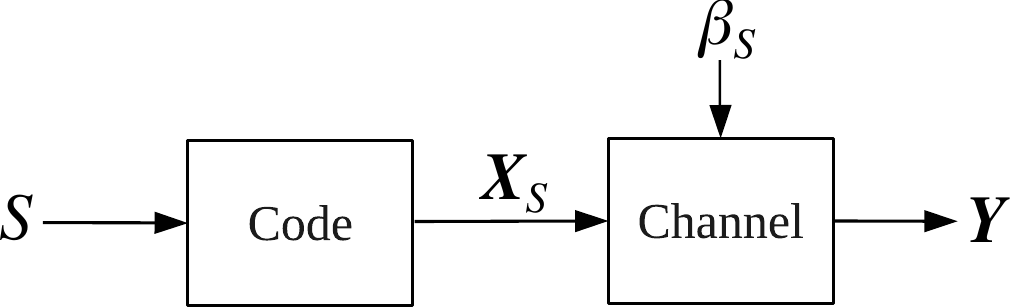}
  \caption{Channel model for structured \& sparse recovery. The support $S$ of the sparse signal is a $K$-set. We map it to a message transmitted through a channel. The encoder encodes $S$ as a matrix, $\X_S$, of size $T \times K$. The coded message $\X_S$ is transmitted through a channel $P(Y|X_S,\bS,S)$ resulting in output $\Y$. As in channel coding, our aim is to identify which message $S$ was transmitted given the channel output $\Y$ and the codebook $\X$. The set $S$ belongs to a family of $K$-sets ${\cal S}$ and can account for combinatorial structural information.}
  \label{fig:channel_model}
\end{figure}

The analysis of the sample complexity is performed by posing this identification problem as an equivalent channel coding problem, as illustrated in Figure \ref{fig:channel_model}.
The sufficiency and necessity results we present in this paper are analogous to the channel coding theorem for memoryless channels \cite{coverbook}. Before we present exact statements of our results, it is useful to mention that these results are of the form
\begin{equation} \label{eq:highlevel}
  T > \max_{\St \subset S} \frac{\log\binom{N-|\St|}{|S \setminus \St|}}{I_\St},
\end{equation}
where $I_\St = \mathrm{ess} \inf_b I(X_\SSt ; Y | X_\St, \bS = b, S)$ is the \emph{worst-case} (w.r.t.\ $\bS$) mutual information between the observation $Y$ and the variables $X_\SSt$ that are in $S$ but not in $\St$, conditioned on variables in $\St$. 
For each subset $\St$ of $S$, this bound can be interpreted as follows: The numerator is the number of bits required to represent all sets $S$ of size $K$ given that its subset $\St$ is already known. In the denominator, the mutual information term represents the uncertainty reduction in the output $Y$ given the remaining input $X_\SSt$ conditioned on a known part of the input $X_\St$, in bits per sample. This term essentially quantifies the ``capacity'' of the observation model $P(Y|X_S,\bS,S)$. Then, the number of samples $T$ should exceed this ratio of total uncertainty to uncertainty reduction per sample for each subset $\St$ to be able to recover $S$ exactly. 

In addition to adapting the channel coding error analysis to the general sparse recovery problem, we also show that the sample complexity is characterized by the \emph{worst-case} rather than average mutual information w.r.t.\ $\bS$, as in \eqref{eq:highlevel} even in our Bayesian setting. We prove that satisying the worst-case bound is both necessary and sufficient for recovery in Sections \ref{sec:necessary} and \ref{sec:sufficient}.

Sparse signal processing models analyzed in this paper have wide applicability. We can account for linear, nonlinear observation models as well as structured settings in a unified manner. Upper and lower bounds for sample complexity for these problems are reduced to computing tight bounds for mutual information (see \eqref{eq:highlevel}).
Below we list some examples of problems which can be formulated in the described framework. Further details concerning the analysis of specific models are provided in Sections \ref{sec:apps_linear} and \ref{sec:apps_nonlinear}. 

\noindent\textbf{Linear channels:}
The prototypical example of a linear channel arises in sparse linear regression \cite{donoho}. Here, the output vector $\Y$ is a linearly transformed $K$-sparse vector $\beta$ with additive noise $\bm{W}$, 
\begin{equation}\label{eq:linCS}
  \Y = \X\beta+\bm{W}.
\end{equation}
We can map this problem to the noisy channel coding (see Fig.~\ref{fig:channel_model}) setup by identifying the support of the sparse vector with the set $S \subset \{1,\ldots,N\}$, the non-zero elements of the support vector with the latent factor ($\bS$), the columns of the matrix $\X$ with codewords that map each ``message'' $S$ into a coded vector $\X_S$. 
It is easy to see that the Markovianity property \eqref{eq:markov_b} holds. The goal is to decode the message $S$ given the codeword matrix $\X$ and channel output $\Y$. 
The channel is linear due to the additive nature of noise. Several variants of this problem can be framed in our setting including sparse or correlated sensing matrices, correlated latent factors, and time-varying latent factors. Correlated sensing matrices arise in many applications due to correlated features and time-varying latent factors can arise in temporal scenarios where latent (nuisance) factors $\bS$ that can vary across different measurements. Note that these variations can be directly mapped to a channel coding formulation by identifying the appropriate probability kernels as in \eqref{eq:markov_b} thus leading to a framework with wide applicability. 

\noindent\textbf{Nonlinear channels:}
In nonlinear channels, the relationship between the observed output $\Y$ and the inputs is not linear. A typical example is the quantized regression problem. Here, the channel output is quantized, namely, $\Y = q(\X \beta + \bm{W})$, where $q(\cdot)$ is a measurement quantizer. Again, we can map this setting to the noisy channel coding framework in an identical manner. 

While the nonlinearity here arises due to quantization, our setup can account for other types of nonlinearities such as Boolean operations.
Group testing \cite{group_testing} is a form of sensing with Boolean arithmetic, where the goal is to identify a set of defective items among a larger set of items.
In an ideal setting, the result of a test (channel output) is positive if and only if the subset contains a positive sample.
The group testing model can also be mapped to the noisy channel coding framework. Specifically, we take a Boolean AND operation over the codewords (rather than a linear transformation) corresponding to the sparse subset $S$.
Note that this is an example of a model where a latent factor $\bS$ does not exist in the observation model or can be considered trivial.
Variants of the nonlinear channel including different noise processes, correlations and time-varying latent factors can also be mapped to our setup.

\noindent\textbf{Noisy features \& missing information:}
Often in many scenarios \cite{loh} the features corresponding to components of the sensing matrix $\X$ are missing. This could occur in medical records where information about some of the tests may not been recorded or conducted. 
We can map this setting as well to our framework and establish corresponding sample complexity bounds.
Specifically, we observe a $T \times N$ matrix $\Z$ instead of $\X$, with the relation
\begin{equation}
Z_i^{(t)} = 
\left\{
\begin{array}{ll}
        X_i^{(t)}, & \mbox{w.p.\  } 1 - \rho\\
        m, & \mbox{w.p.\  } \rho
\end{array}\right. \quad \forall i, t
\notag\\
\end{equation}
i.e., we observe a version of the feature matrix which may have missing entries (denoted by $m$) with probability $\rho$, independently for each entry. Note that $m$ can take any value as long as there is no ambiguity whether the realization is missing or not, e.g., $m=0$ would be valid for continuous variables where the variable taking value $0$ has zero probability. We also remark that if a problem satisfies assumption \eqref{eq:markov} with IID variables $X$, the same problem with missing features also satisfies the assumption with variables $Z$. Interestingly, our analysis shows that the sample complexity, $T_\text{miss}$ for problems with missing features is related to the sample complexity, $T$, of the fully observed case with no missing features by the simple inequality
\[ T_\text{miss} \geq \frac{T}{1-\rho}. \]

\noindent\textbf{Structured sparse settings:}
In addition to linear and nonlinear settings and their variants we can also deal with structural information on the support set $S$ such as group sparsity and subgraph connectivity. 

For instance, consider the multiple linear regression problem with sparse vectors sharing the same support \cite{negahban,dcs}. The message set can be viewed as belonging to a subset of the family of binary matrices, namely, $S \in {\cal S} \subset \{0,1\}^{N \times R}$ for $R$ problems where the collection ${\cal S}$ is a structured family of binary matrices. In the simple case the structural information arises from the fact that the collection ${\cal S}$ is a finite collection of binary matrices such that all its columns have identical support. This problem can be viewed as $R$ linear regression problems and can be expressed for linear channels as $\Y_{\{r\}} = \X_{\{r\}} \beta_{\{r\}} + \bm{W}_{\{r\}},~r=1,\ldots,R$.
For each $r$, $\beta_{\{r\}} \in \mathbb{R}^N$ is a $K$-sparse vector, $\X_{\{r\}} \in \mathbb{R}^{T \times N}$, $\Y_{\{r\}} \in \mathbb{R}^T$ and the relation between tasks is that $\beta_{\{r\}}, r = 1, \ldots, R$ have the same support. One could also extend this framework to other group-sparse problems by imposing various constraints on the collection ${\cal S}$.

An interesting case is where the structural information can be encoded with respect to an $N$-node graph $G=(V,E)$. Here, we can consider the collection ${\cal S}$ as the family of all \emph{connected} subgraphs of size $K$. Thus $S$ is a $K$-set of $K$ nodes whose induced subgraph is connected. These are problems that can arise in many interesting scenarios \cite{nips14,aistats14} such as disease outbreak detection, medical imaging and inverse problems, where the underlying signal must satisfy connectivity constraints. In essence our methods apply to these cases as well. Indeed, our sample complexity expressions in these structured cases reduce to
\begin{equation} \label{eq:highlevel_struc}
  T > \max_{\St \subset S} \frac{\log |{\cal S}_{\St}|}{I_\St},
\end{equation}
where ${\cal S}_{\St} = \{ S' \in {\cal S} : S' \supset \St \}$. Intuitively, ${\cal S}_{\St}$ is the set of all structures that are consistent with the partially recovered set $\St$. 
This generalization reveals a \emph{key aspect} of our formulation. Specifically, we can think of inference problems where the goal is to decode elements that belong to some combinatorial structure. We know how the combinatorial structure manifests as encoded message $\X_S$ for feasible $S \in {\cal S}$. Our bound shows that only the numerator changes (see \eqref{eq:highlevel} and \eqref{eq:highlevel_struc}). Intuitively this modification accounts for the number of feasible $K$-sets.

\section{Related Work and Summary of Contributions}
The dominant stream of research on sparse recovery focuses on the \emph{linear compressive sensing (CS) model}, often with mean-squared \emph{estimation} of the sparse vector $\beta$ in \eqref{eq:linCS} with sub-Gaussian assumptions on the variables $X$.

While the linear model is well-studied, research on information-theoretic limits of general nonlinear models is in its early stages and primarily limited to specific models such as Boolean group testing \cite{group_testing} and quantized compressive sensing \cite{1bit,planvershynin}. In this paper we seek to understand the fundamental information-theoretic limits for \emph{generalized models} of (both linear and nonlinear) sparse signal processing using a unifying framework that draws an analogy between channel coding and the problem of sparse support recovery. Our main results on mutual information characterizations of generalized models, stated in Theorems \ref{thm:eps_lb} and \ref{thm:latent_ub} are then used to derive both necessary and sufficient conditions for specific models, including the linear model (Section \ref{subsec:lcs}), 1-bit CS (Section \ref{subsec:1bit}), models with missing features (Section \ref{subsec:missing}), group testing (Section \ref{subsec:group_testing}). The derived bounds are often shown to match existing sufficiency and necessity conditions, or provide tighter bounds for some of these models. As such, our abstract bounds provide a tight characterization of sample complexity and the gap between our upper and lower bounds in some applications results from the difficulty in finding tight upper and lower bounds for mutual information. Below we provide a brief discussion of related prior work for both linear and nonlinear sparsity-based models, then provide a summary of contributions and contrast to prior work.

\subsection{Linear Model} \label{sec:linmod}

This literature can be broadly classified into two main categories. The first category is primarily focused on the analysis of \emph{computationally tractable} algorithms for support recovery and deriving both sufficient and necessary conditions for these algorithms (see \cite{wainwright-lasso,goyal,candesplan,zhaosrvit}).
The second category, which is more relevant to our work, is focused on the complementary task of characterizing the fundamental limits for sparse recovery regardless of the computational complexity of the used algorithms. The importance of this line of work lies in assessing the behavior of tractable algorithms and uncovering the performance gaps from the information-theoretic limits in various regimes. 

\noindent
\textbf{Necessary condition \& Fano's inequality:}
In this line of work \cite{akcakaya,wainwright,wang,shuchin,reeves2} lower bounds for sample complexity are derived by invoking various forms of Fano's inequality for different scenarios including Gaussian ensembles, sparse sensing matrices, high/low SNRs, and linear or sublinear sparsity regimes.
Our lower bound follows the proof of Fano's inequality. However, the main difference between these existing bounds and ours is in how we account for the latent factor $\bS$. It turns out that if one were to apply the standard forms of the Fano's inequality as in the existing literature it results in ``averaging'' out the effect of latent factors leading to standard mutual information expressions between message and output alphabets. Nevertheless, these resulting bounds are too loose. Intuitively, $\bS$ can be thought of as the unknown state of a compound channel and if bad realizations have non-zero probability, this must factor into the lower bound. Using this intuition, we derive a novel lower bound for the sample complexity. This lower bound surprisingly is simple and explicitly shows that the worst-case conditional mutual information over $\bS$ realizations quantifies sample complexity.

\noindent
\textbf{Sufficiency bounds with ML decoder---bypassing $\bS$ estimation:}
\cite{shuchin,Rad-11,yhkim} derive sufficient conditions for support recovery for structured and unstructured Gaussian ensembles using an exhaustive search decoder, which searches for the best fit among different choices of $S$ and $\bS$.  
Intuitively, this setup amounts to explicit estimation of the latent variable $\bS$ in the process of identifying $S$.

In contrast, we employ an ML decoder where the latent variable $\bS$ is part of the channel and plays the role of a latent variable. 
We bypass the $\bS$ estimation step and focus our attention on the recovery of the discrete combinatorial component $S$ (the channel input), while the effects of the support coefficients $\bS$ with prior $P(\bS | S)$ are incorporated to the channel model $P(Y | X, S)$ in a Bayesian framework such that
\begin{equation*}
  P(Y | X, S) \!=\! P(Y | X_S, S) \!=\! \int P(Y|X_S,\bS,S) p(\bS) \dx \bS.
\end{equation*}
Our resulting bounds explicitly demonstrate that identifying the support dominates the sample complexity for sparse recovery. This is intuitive because one can reliably estimate the underlying latent variable $\bS$ using least-square estimates or other variants once the support is known.

In this context our approach is closely related to that of \cite{tang}, where they also bypass $\bS$ estimation step. They formulate support recovery as a multiple hypothesis testing problem with $\binom{N}{K}$ hypotheses, each hypothesis corresponding to one possible support set of size $K$. They derive lower and upper bounds on the decoding error probability for a multiple measurement model (assuming the availability of multiple temporal samples) using Fano's inequality and Chernoff bound. The performance analysis is derived conditioned on a specific realization of the measurement matrix $\X$ and $\beta$ is Gaussian. In contrast to our work, this paper is focused on the scaling of the number of temporal samples, but not on the scaling of $T,N,K$. Furthermore, the paper exploits the additive Gaussian noise and linearity of the channel structure for deriving the bounds. In contrast, our method is general and extends to nonlinear channels as well.

\subsection{Nonlinear Models, MAC Channels \& Unifying Framework}

In the sparse recovery literature there have been a few works that have focused on nonlinear models such as 1-bit quantization \cite{1bit,jacques,planvershynin} and group testing \cite{jaggi}. Nevertheless, the focus of these works has been on computationally tractable algorithms.
Our approach is more closely related to the channel coding framework of \cite{group_testing,malyutov,malyutov1,malyutov2,malyutov3,dyachkov_lectures} in the context of group testing.
More generally, our approach bears some similarities to multi-user multi-access communication (MAC) systems literature \cite{coverbook}. 

While our approach is inspired by this literature, there are major differences in terms of scope, results, and proof techniques. As such, our setup does not directly fall into the MAC setting because our sparse recovery problem is in essence an on-off MAC channel where $K$ out of $N$ users can be active. Unlike the multi-user setting where all the users are known and the goal is to decode each user's codeword, we do not know which of the $K$ users are active. Alternatively, as observed in \cite{yhkim,group_testing} unlike MAC where each user chooses from a different codebook, here all users choose from the same codebook. Furthermore, unlike the MAC setting where the channel gains are assumed to be known at the decoder, we do not know the latent factors in our setting. 

On the other hand the group-testing approaches are not directly applicable as well. One issue is that they are not tractable for general discrete and continuous alphabets considered in this paper. Second, group-testing does not involve a compound channel. Consequently, the lower bounds in \cite{group_testing} are too loose and latent factors play a fundamental role as seen from our lower bounds (see Section \ref{sec:linmod}). Third, from a technical perspective, unlike \cite{group_testing} our new analysis is not based on bounds of the second derivative of the error exponent, which turns out to be intractable in many cases. Instead, we develop novel analysis techniques that exploit the equicontinuity of the error exponent function. The new results do not require problem specific computation other than satisfying generic regularity conditions (see Definition \ref{def:F_N}), which are shown to be easily satisfied for a wide range of models as derived in the applications section (c.f.\ Sections \ref{sec:apps_linear} and \ref{sec:apps_nonlinear}). To the best of our knowledge, this paper provides the first unifying information-theoretic characterization for support recovery for sparse and structured signal processing models with linear and nonlinear observations. Our approach unifies the different sparse models based on the conditional independence assumption in \eqref{eq:markov}, and simple mutual information expressions (Theorems \ref{thm:eps_lb} and \ref{thm:latent_ub}) are shown to provide an exact characterization of the sample complexity for such models.

\subsection{Key Insights \& New Bounds}
We now present three key insights and describe new bounds we obtain with our approach.

\noindent
\textbf{Necessity of $\beta_{\min}$ assumption:}
Support pattern recovery guarantees in \cite{akcakaya,wainwright,wang,shuchin,reeves2} for linear channels require that the minimum value of $\beta_S$ is bounded from below by $\beta_{\min}$ but do not provide an explicit justification. While this is intuitive because these works rely on estimating $\beta_S$ for support recovery, it is unclear whether this is fundamental. Our sample complexity bounds provide an information-theoretic explanation for the necessity of this assumption for recovery with small probability of error. Notably our analysis considers the average error for Bayesian $\beta_S$, not worst-case for fixed and unknown $\beta_S$. 

\noindent
\textbf{Role of structure:}
A key contribution of our formulation is that it reveals the role of structure in inference problems in an explicit manner. In particular, we see that if our object of inference is to decode elements from some combinatorial structure, its role is limited to the cardinality of this structure (numerator in \eqref{eq:highlevel} and \eqref{eq:highlevel_struc}) with the mutual information expression remaining unchanged.

\noindent
\textbf{Role of sensing matrices, missing features, \& latent variables:}
Because of the simplicity of our expressions we can \emph{explicitly} study the impact of correlations in feature components of the sensing matrix (higher correlations leading to poorer sample complexity), correlations in latent variables (higher correlations lead to better sample complexity) and missing features.

\noindent
\textbf{New bounds and improvements over existing bounds:}
Our approach enables us to obtain new necessary and sufficient conditions for support recovery not considered before in the literature for various sparse models. In addition, we also improve upon existing bounds in many cases. For instance, 
for group testing, we are able to remove the additional polylog factors in $K$ arising in \cite{group_testing} (see Theorem \ref{thm:group_testing}) leading to tight upper and lower bounds for the problem.
We get sharper bounds for the missing features model with linear observations that improve over the bounds in \cite{caramanis} as shown in Theorem \ref{thm:missing_cs} and Remark \ref{remark:missing_cs}.
We also obtain tighter bounds for multilinear regression in Section \ref{subsec:multi}.
Some of these bounds are summarized in Table \ref{table:apps} and other cases are described in Sections \ref{sec:apps_linear} and \ref{sec:apps_nonlinear}. 
The generality of this framework and the simple characterization in terms of mutual information expressions enable applicability to other models not considered in this paper that could potentially lead to new bounds on sample complexity.

\begin{table}[tbp]
  \renewcommand{\arraystretch}{1.5}
  \caption{Sample complexity bounds derived through unifying results for the general model and specific applications for exact support recovery. Results for applications are presented and proved in the corresponding subsections in Sections \ref{sec:apps_linear} and \ref{sec:apps_nonlinear}.}
  \centering
  \setlength{\tabcolsep}{2pt}
  \begin{tabular}{| c | c | c |}
    \hline
    \textbf{Model} & \textbf{Sufficient conditions for $\Pr[\mathrm{error}] \to 0$} & \textbf{Necessary conditions for $\Pr[\mathrm{error}] \to 0$} \\ \hline
    General model  & $T > C \max\limits_{\St \subset S} \frac{\log\binom{N-K}{K-|\St|} + H_\frac{1}{2}(\bS)}{I_\St}$ & $T \geq \max\limits_{\St \subset S} \frac{\log\binom{N-K}{K-|\St|}}{I_\St}$ \bigstrut \\ \hline
    \specialcell{Sparse linear regression\footnotemark[1] with IID \\ or correlated $\bS$, $\bmin = \Theta(1/K)$}
    & $T = \Omega(K \log N)$ for $K = O(N)$\footnotemark[2]
    & $T = \Omega(K \log N)$ for $K = O(N)$ \\ \hline
    \multirow{2}{*}{\specialcell{Sparse linear regression\footnotemark[1] with IID \\ or correlated $\bS$, $\bmin = \Theta(\log K/K)$}}
    & $T \!=\! \Omega\left(\max\left\{\frac{K \log N}{\log K}, \frac{K \log(N/K)}{\log\log K} \right\}\right)$, $K \!=\! o(N)$\footnotemark[2]
    & $T \!=\! \Omega\left(\max\left\{\frac{K \log N}{\log K}, \frac{K \log(N/K)}{\log\log K} \right\}\right)$, $K \!=\! o(N)$ \\ \cline{2-3}
    & $T = \Omega(N)$ for $K = \Theta(N)$\footnotemark[3] & $T = \Omega(N)$ for $K = \Theta(N)$ \\ \hline
    Multivariate regression with $R$ problems & $T \geq \frac{T_\mathrm{single}}{R}$ & $T \geq \frac{T_\mathrm{single}}{R}$ \\ \hline
    Binary regression & $T = \Omega(K \log N)$ for $K = \Theta(1)$ & $T = \Omega(K \log(N/K))$ for $K = O(N)$ \\ \hline
    Group testing & $T = \Omega(K \log N)$ for $K = O(N)$ & $T = \Omega(K \log(N/K))$ for $K = O(N)$ \\ \hline
    General models with missing data w.p.\ $\rho$ & -- & $T \geq \frac{T_\mathrm{full}}{1-\rho}$ \\ \hline
    \specialcell{Sparse linear regression\footnotemark[1] with missing \\ data w.p.\ $\rho$ and $\bmin = \Theta(1/K)$}
      & \specialcell{$T = \Omega\left(\frac{K \log N}{\log\left(1 + \frac{1-\rho}{1+\rho}\right)}\right)$  for $K = O(N)$\footnotemark[2]}
      & \specialcell{$T = \Omega\left( \frac{K \log N}{1-\rho} \right)$ for $K = O(N)$} \\ \hline
  \end{tabular}
  \label{table:apps}
\end{table}

\footnotetext[1]{Using the setup of \cite{wang} and \cite{Rad-11} as described in Section \ref{subsec:lcs}.}
\footnotetext[2]{Holds as written for highly correlated $\bS$ and exact recovery, holds for $K = O(N/\log N)$ and for recovery with a vanishing fraction of support errors for IID $\bS$.}
\footnotetext[3]{Holds for highly correlated $\bS$.}

Preliminary results for the setting considered herein and extensions to other settings have previously been presented in \cite{ssp,icassp,itw,aistats,isit}. \cite{ssp} and more recently \cite{aistats} considered the setting with dependent variables, while \cite{isit} extended the lower bound analysis to sparse recovery with adaptive measurements.

We describe the organization of the paper. In Section \ref{sec:setup} we introduce our notation and provide a formal description of the problem. In Section \ref{sec:necessary} we state necessary conditions on the number of samples required for recovery, while in Section \ref{sec:sufficient} we state sufficient conditions for recovery. Discussions about the conditions derived in Sections \ref{sec:necessary} and \ref{sec:sufficient} are presented in Section \ref{sec:discussion}. Applications are considered in Sections \ref{sec:apps_linear} and \ref{sec:apps_nonlinear}, including bounds for sparse linear regression, group testing models, and models with missing data. We summarize our results in Section \ref{sec:conclusions}.

\section{Problem Setup}
\label{sec:setup}

\textbf{Notation.}
We use upper case letters to denote random variables, vectors and matrices, and we use lower case letters to denote realizations of scalars, vectors and matrices. Calligraphic letters are used to denote sets or collection of sets, usually sample spaces for random quantities. Subscripts are used for column indexing and superscripts with parentheses are used for row indexing in vectors and matrices. Bold characters denote multiple samples jointly for both random variables and realizations and specifically denote $T$ samples unless otherwise specified. Subscripting with a set $S$ implies the selection of columns with indices in $S$. Table \ref{table:notation} provides a reference and further details on the used notation. The transpose of a vector or matrix $a$ is denoted by $a^\top$. $\log$ is used to denote the natural logarithm and entropic expressions are defined using the natural logarithm, however results can be converted to other logarithmic bases w.l.o.g., such as base 2 used in \cite{group_testing}. The symbol $\subseteq$ is used to denote subsets, while $\subset$ is used to denote proper subsets.

Without loss of generality, we use notation for discrete variables and observations throughout the paper, i.e.\ sums over the possible realizations of random variables, entropy and mutual information definitions for discrete random variables etc.\ The notation is easily generalized to the continuous case by simply replacing the related sums with appropriate integrals and (conditional) entropy expressions with (conditional) differential entropy, excepting sections that deal specifically with the extension from discrete to continuous variables, such as the proof of Lemma \ref{lemma:P(Ei)} in Section \ref{subsec:continuous}.

\begin{table}[b!]
  \renewcommand{\arraystretch}{1.2}
  \caption{Reference for notation used}
  \centering
  \begin{tabular}{| l | c | c |}
    \hline
    & \textbf{Random quantities} & \textbf{Realizations} \\ \hline
    Variables & $X_1, \ldots, X_N$ & $x_1, \ldots, x_N$ \\
    $1\!\times\! N$ random vector & $X = (X_1, \ldots, X_N)$ & $x = (x_1, \ldots, x_N)$ \\
    $1 \!\times\! |S|$ random vector & $X_S$ & $x_S$ \\
    $T \!\times\! N$ random matrix & $\X$  & $\x$ \\
    $t$-th row of $\X$ & $X^{(t)}$ & $x^{(t)}$ \\
    $n$-th column of $\X$ & $\X_n$ & $\x_n$ \\
    $n$-th elt.\ of $t$-th row & $X_n^{(t)}$ & $x_n^{(t)}$ \\
    $T \!\times\! |S|$ sub-matrix & $\X_S$ & $\x_S$ \\
    Observation & $Y$ & $y$ \\
    $T \!\times\! 1$ observation vector & $\Y$ & $\y$ \\
    $t$-th element of $\Y$ & $Y^{(t)}$ & $y^{(t)}$ \\
    \hline
  \end{tabular}
  \label{table:notation}
\end{table}

\textbf{Variables.} 
We let $X = (X_1, X_2, \ldots, X_N) \in {\cal X}^N$ denote a set of IID random variables with a joint probability distribution $Q(X)$. We specifically consider discrete spaces ${\cal X}$ or finite-dimensional real coordinate spaces $\mathbb{R}^d$ in our results. To avoid cumbersome notation and simplify the expressions, we do not use subscript indexing on $Q(\cdot)$ to denote the random variables since the distribution is determined solely by the number of variables indexed.

\textbf{Candidate sets.} 
We index the different sets of size $K$ as $S_{\omega}$ with index $\omega$, so that $S_{\omega}$ is a set of $K$ indices corresponding to the $\omega$-th set of variables. Since there are $N$ variables in total, there are $\binom{N}{K}$ such sets, therefore $\omega\in{\cal I} \triangleq \left\{1,2,\ldots\binom{N}{K}\right\}$. This index set is isomorphic to ${\cal S}$ for unstructured problems. We assume the true set $S = \So$ for some $\omega \in {\cal I}$.

\textbf{Latent observation parameters.} 
We consider an observation model that is not fully deterministic and known, but depends on a latent variable $\bS \in {\cal B}^K$. 
We assume $\bS$ is independent of variables $X$ and has a prior distribution $P(\bS | S)$, which is independent of $S$ and symmetric (permutation invariant). We further assume that $\beta_k$ for $k \in S$ has finite R\'enyi entropy of order $1/2$, i.e.\ $H_{\frac{1}{2}}(\beta_k) < \infty$ and also that $H_\frac{1}{2}(\bS) = O(K)$.

\textbf{Observations.} 
We let $Y \in {\cal Y}$ denote an observation or outcome, which depends only on a small subset of variables $S\subset \{ 1,\ldots, N\}$ of known cardinality $|S| = K$ where $K \ll N$. In particular, $Y$ is conditionally independent of the variables given the subset of variables indexed by the index set $S$, as in \eqref{eq:markov}, i.e., $P(Y | X, S) = P(Y | X_S, S)$, where $X_S = \{X_k\}_{k\in S}$ is the subset of variables indexed by the set $S$. 
The outcomes depend on $X_S$ (and $\bS$ if exists) and are generated according to the model $P(Y | X_S, S)$ (or $P(Y | X_S, \bS, S)$). 

We further assume that the observation model is independent of the ordering of variables in $S$ such that
\[ P(Y | X_S = x_s, S) = P(Y | X_S = x_{\pi(S)}, S) \]
for any permutation mapping $\pi$, which allows us to work with sets (that are unordered) rather than sequences of indices. Also, the observation model does not depend on $S$ except through $X_S$, i.e.,
\[ P(Y | X_\So = x, \So) = P(Y | X_{S_{\hat{\omega}}} = x, S_{\hat{\omega}}) \]
for any $x \in {\cal X}^K$, $\omega, \hat{\omega} \in {\cal I}$.

We use the lower-case $p(\,\cdot\, | \,\cdot\,) = P(\,\cdot\, | \,\cdot\, , S)$ notation as a shorthand for the conditional distribution given the true subset of variables $S$. For instance, with this notation we have $p(Y | X_S) = P(Y | X_S, S)$, $p(Y | X_S, \bS) = P(Y | X_S, \bS, S)$, $p(\bS) = P(\bS | S)$ etc. Whenever we need to distinguish between the outcome distribution conditioned on different sets of variables, we use $p_\omega(\,\cdot\, | \,\cdot\,) = P(\,\cdot\, | \,\cdot\, , \So)$ notation, to emphasize that the conditional distribution is conditioned on the given variables, assuming the true set $S$ is $\So$.

We observe the realizations $(\x, \y)$ of $T$ variable-outcome pairs $(\X,\Y)$ with each sample realization $(x^{(t)}, y^{(t)})$ of $(X^{(t)},Y^{(t)})$, $t=1,2,\ldots,T$. The variables $X^{(t)}$ are distributed IID across $t=1,\ldots,T$. However, if $\bS$ exists, the outcomes $Y^{(t)}$ are independent for different $n$ only when conditioned on $\bS$. Our goal is to identify the set $S$ from the data samples and the associated outcomes $(\x,\y)$, with an arbitrarily small average error probability. 

\textbf{Decoder and probability of error.} 
We let $\hat{S}(\X,\Y)$ denote an estimate of the set $S$, which is random due to the randomness in $S$, $\X$ and $\Y$. 
We further let $P(E)$ denote the average probability of error, averaged over all sets $S$ of size $K$, realizations of variables $\X$ and outcomes $\Y$, i.e., 
\begin{equation*}
  P(E) = \Pr[\hat{S}(\X,\Y)\ne S] = \sum_{\omega \in {\cal I}} P(\omega)\Pr[\hat{S}(\X,\Y)\ne \So | \So].
\end{equation*}

\textbf{Scaling variables and asymptotics.} 
We let $N \in \mathbb{N}$, $K \triangleq K(N) \in \mathbb{N}$ be a function of $N$ such that $1 \leq K < N/2$ and $T \triangleq T(K,N) \in \mathbb{N}$ be a function of both $K$ and $N$. Note that $K$ can be a constant function in which case it does not depend on $N$. 
For asymptotic statements, we consider $N \to \infty$ and $K$ and $T$ scale as defined functions of $N$. We formally define the notion of \emph{sufficient} and \emph{necessary} conditions for recovery below (Definition \ref{def:suff_nec}).

\textbf{Conditional entropic quantities.}
We occasionally use conditional entropy and mutual information expressions conditioned on a fixed value or on a fixed set. For two random variables $U \in {\cal U}$ and $V \in {\cal V}$, we use the notation $H(U | V=v) = -\sum_u p(u|v) \log p(u|v)$ to denote the conditional entropy of $U$ conditioned on fixed $V=v$. For a measurable subset ${\cal V}' \subseteq {\cal V}$ of the space of realizations of $V$, $H(U | V \in {\cal V}') = -\frac{1}{P({\cal V}')} \sum_{v \in {\cal V}'} p(v) \sum_u p(u|v) \log p(u|v)$ denotes the conditional entropy of $U$ conditioned on $V$ being restricted to set ${\cal V}'$. Note that this is equivalent to the (average) conditional entropy $H(U | V)$ when ${\cal V}' = {\cal V}$. The differential entropic definitions for continuous variables follow by replacing sums with integrals, and conditional mutual information terms follow from the entropy definitions above.

\begin{definition}
\label{def:suff_nec}
  For a function $g \triangleq g(T,K,N)$, we say an inequality $g \geq 1$ (or $g > 1$) is a \emph{sufficient condition} for recovery if there exists a sequence of decoders $\hat{S}_N(\X, \Y)$ such that $\lim_{N \to \infty} P(E) = \lim_{N \to \infty} \Pr[\hat{S}_N(\X, \Y) \neq S] = 0$ for $g \geq 1$ (or $g > 1$) for sufficiently large $N$, i.e., for any $\epsilon > 0$, there exists $N_\epsilon$ such that for all $N > N_\epsilon$, $g \geq 1$ (or $g > 1$) implies $P(E) < \epsilon$. 

  Conversely, we say an inequality $g \geq 1$ (or $g > 1$) is a \emph{necessary condition} for recovery if when the inequality is violated, $\lim_{N \to \infty} P(E) > 0$ for any sequence of decoders.
\end{definition}

To recap, we formally list the main assumptions that we require for the analysis in the work below.
\begin{enumerate}[({A}1)]
  \item \textbf{Equiprobable support:} Any set $\So \subset \{1,\ldots,N\}$ with $K$ elements is equally likely \emph{a priori} to be the salient set.
  We assume we have no prior knowledge of the salient set $S$ among the $\binom{N}{K}$ possible sets.
  \item \textbf{Conditional independence and observation symmetry:} The observation/outcome $Y$ is conditionally independent of other variables given $X_S$ (variables with indices in $S$), i.e., $P(Y | X, S) = P(Y | X_S, S)$. For any permutation mapping $\pi$, $P(Y | X_S = x_s, S) = P(Y | X_S = x_{\pi(S)}, S)$, i.e., the observations are independent of the ordering of variables. This is not a restrictive assumption since the asymmetry w.r.t.\ the indices can be usually incorporated into $\bS$. In other words, the symmetry is assumed for the observation model when averaged over $\bS$. 
  We further assume that the observation model does not depend on $S$ except through $X_S$, i.e., $P(Y | X_\So = x, \So) = P(Y | X_{S_{\hat{\omega}}} = x, S_{\hat{\omega}})$ for any $x \in {\cal X}^K$, $\omega, \hat{\omega} \in {\cal I}$.
\item \textbf{IID variables:} The variables $X_1, \ldots, X_N$ are independent and identically distributed. While the independence assumption is not valid for all sparse recovery problems, many problems of interest can be analyzed within the IID framework, as in Sections \ref{sec:apps_linear} and \ref{sec:apps_nonlinear}.
  \item \textbf{IID samples:} The variables $X^{(t)}$ are distributed IID across $t=1,\ldots,T$.
  \item \textbf{Memoryless observations:} Each observation $Y^{(t)}$ at sample $t$ is independent of $X^{(t')}$ conditioned on $X^{(t)}$.
\end{enumerate}

In the next three sections, we state and prove necessary and sufficient conditions for the recovery of the salient set $S$ with an arbitrarily small average error probability. We start with deriving a necessity bound in Section \ref{sec:necessary}, then we state a corresponding sufficiency bound in Section \ref{sec:sufficient}. We discuss the extension of the sufficiency bound to scaling models in Section \ref{subsec:scaling} and we conclude with remarks on the derived results in Section \ref{sec:discussion}.

\section{Necessary Conditions for Recovery}
\label{sec:necessary}

In this section, we derive lower bounds on the required number of samples using Fano's inequality. First we formally define the following mutual information-related quantities, which we will show to be crucial in quantifying the sample complexity in Sections \ref{sec:necessary} and \ref{sec:sufficient}.

For a proper subset $\St$ of $S$, the \emph{conditional mutual information} conditioned on fixed $\bS = b \in {\cal B}^K$ is\footnote{\label{fn:thm}We refer the reader to Section \ref{sec:setup} for the formal definition of conditional entropic quantities. 
  We also note that it is sufficient to compute $I_\St(b)$ for one value of $\om$ (e.g.\ $S_1$) instead of averaging over all possible $S$, since the conditional mutual information expressions are identical due to our symmetry assumptions on the variable distribution and the observation model. Similarly, the bound need only be computed for one proper subset $\St$ for each $|\St| \in \{0,\ldots,K-1\}$, since our assumptions ensure that the mutual information is identical for all partitions $(\SSt, \St)$.
}
\begin{equation}\label{def:cond_MI}
  I_\St(b) = I(X_\SSt ; Y | X_\St, \bS = b, S),
\end{equation}
the \emph{average conditional mutual information} conditioned on $\bS \in B$ for a measurable subset $B \subseteq {\cal B}^K$ is 
\begin{equation}\label{def:avg_MI}
  I_\St(B) = I(X_\SSt ; Y | X_\St, \bS \in B, S), 
\end{equation}
and the \emph{$\varepsilon$-worst-case conditional mutual information} w.r.t.\ $\bS$ is 
\begin{equation}\label{def:worst_MI_q}
  I_{\St,\varepsilon} = \sup \{\alpha \in \mathbb{R}^+: \Pr[ b \in {\cal B}^K: I_\St(b) < \alpha ] \leq \varepsilon \}.
\end{equation}
Note that for $\varepsilon = 0$, $I_{\St,\varepsilon}$ reduces to the essential infimum of $I_\St(\cdot)$, $\mathrm{ess } \inf_{b \in {\cal B}^K} I_\St(b) = \sup \{\alpha \in \mathbb{R}^+: \Pr[ b \in {\cal B}^K: I_\St(b) < \alpha ] = 0 \}$.

We now state the following theorem as a tight necessity bound for recovery with an arbitrarily small probability of error.

\begin{theorem}\label{thm:eps_lb}
  For any $\frac{2}{\log(N-K+1)} \leq \varepsilon \leq 1$, if 
  \begin{equation}\label{eq:eps_lb}
    T < (1 - \varepsilon) \max_{\St \subset S} \frac{\log\binom{N-|\St|}{K-|\St|}}{I_{\St,\varepsilon}},
  \end{equation}
  then $P(E) > \frac{\varepsilon^2}{2}$.
\end{theorem}

This necessary condition implies that a worst-case condition on $\bS$ has to be satisfied for recovery with small average error probability, which we will show to be consistent with the upper bounds we prove in the following section. 

\begin{remark}\label{remark:zero_lb}
  A necessary condition for recovery with \emph{zero}-error in the limit can be easily recovered from this theorem by considering an arbitrarily small constant $\varepsilon > 0$. Then the mutual information term $I_{\St,\varepsilon}$ represents the worst-case mutual information barring subsets of ${\cal B}^K$ with an arbitrarily small probability and thus can be considered $I_{\St,0}$ for most problems. Therefore we essentially have that if
  \[ T < \max_{\St \subset S} \frac{\log\binom{N-|\St|}{K-|\St|}}{I_{\St,0}}, \]
  then $\lim_{N \to \infty} P(E) > 0$.
\end{remark}

To prove Theorem \ref{thm:eps_lb}, we first state and prove the following lemma that lower bounds $P(E)$ for any subset $B \subseteq {\cal B}^K$ and any tuple $(T,K,N)$.

\begin{lemma}\label{lemma:beta_lb}
  \begin{equation}\label{eq:beta_lb}
    P(E) \geq \Pr[ \bS \in B ] \left( 1 - \frac{T I_\St(B) + 1}{\log\binom{N-|\St|}{K-|\St|}} \right).
  \end{equation}
\end{lemma}

\begin{proof}
  Let the true set $S = \So$ for some $\omega \in {\cal I}$ and suppose a proper subset of elements of $\So$ is revealed, denoted by $\St$. We define the estimate of $\omega$ to be $\hat{\omega}=g(\X,\Y)$ and the probability of error in the estimation $P_e = P(E) = \Pr[\hat{\omega}\ne\omega]$. We analyze the probability of error conditioned on the event $\bS \in B$, which we denote with $P(E | B) = \Pr[\hat{\omega}\ne\omega | \bS \in B]$. We note that for continuous variables and/or observations we can replace the (conditional) entropy expressions with differential (conditional) entropy, as we noted in the problem setup.

  Conditioning on $\bS \in B$ using the notation established in Section \ref{sec:setup}, we can write
  \[ H(\om | \X, \Y, \St, \bS \in B) + H(E | \om, \X, \Y, \St, \bS \in B) = H(E | \X, \Y, \St, \bS \in B) + H(\om | E, \X, \Y, \St, \bS \in B), \]
  where $H(E | \om, \X, \Y, \St, \bS \in B) = 0$ since $E$ is completely determined by $\om$ and $\omh$ which is a function of $\X, \Y$.
  Upper bounding $H(E | \X, \Y, \St, \bS \in B) \leq 1$ and expanding $H(\om | E, \X, \Y, \St, \bS \in B)$ for $E = 0$ (denoting $\om = \omh$) and $E = 1$ ($\om \ne \omh$), we have
  \[ H(\om | \X, \Y, \St, \bS \in B) \leq 1 + P(E | B) \log\binom{N-|\St|}{K-|\St|}, \]
  since $H(\om | E = 0, \X, \Y, \St, \bS \in B) = 0$ and $H(\om | E = 1, \X, \Y, \St, \bS \in B) \leq H(\om | \St) = \log\binom{N-|\St|}{K-|\St|}$.

  For $H(\om | \X, \Y, \St, \bS \in B)$, we also have a lower bound from the following chain of inequalities:
  \begin{align*}
    H(\omega | \Y,\X,\St, \bS \in B) & = H(\omega | \St, \bS \in B) - I(\omega ; \Y,\X | \St, \bS \in B) \\
    & = H(\omega | \St, \bS \in B) - I(\omega ; \X | \St, \bS \in B) - I(\omega ; \Y | \X, \St, \bS \in B) \\
    & \stackrel{(a)}{=} H(\omega | \St, \bS \in B) - I(\omega ; \Y | \X ,\St, \bS \in B) \\
    & \stackrel{(b)}{=} H(\omega | \St, \bS \in B) -  (H(\Y | \X, \St, \bS \in B) -  H(\Y | \X, \omega, \bS \in B)) \\
    & \stackrel{(c)}{\geq} H(\omega | \St, \bS \in B) - (H(\Y | \X_{\St}, \St, \bS \in B) -  H(\Y | \X_\So, \omega, \bS \in B)) \\
    & \stackrel{(d)}{=} H(\omega | \St, \bS \in B) - I(\X_{\SSt}; \Y | \X_{\St}, \So, \bS \in B) \\
    & = \log\binom{N-|\St|}{K-|\St|} - I(\X_\SSt ; \Y | \X_\St, \So, \bS \in B),
  \end{align*}
  where (a) follows from the fact that $\X$ is independent of $\St$ and $\omega$; (b) follows from the fact that conditioning on $\omega$ includes conditioning on $\St$; (c) follows from the fact that conditioning on less variables increases entropy and for the second term that $\Y$ depends on $\omega$ only through $\X_\So$; (d) follows by noting that $\SSt$ does not give any additional information about $\Y$ when $\X_\SSt$ is marginalized, because of our assumption that the observation model is independent of the indices themselves except through the variables and we have symmetrically distributed variables, and therefore $H(\Y | \X_\St, \St, \bS \in B) = H(\Y | \X_\St, \So, \bS \in B)$.

  From the upper and lower bounds derived on $H(\om | \X, \Y, \St, \bS \in B)$, we then have the inequality
  \begin{equation}\label{eq:Peb_lb}
    P(E | B) \geq 1 - \frac{I(\X_{\SSt}; \Y | \X_{\St}, \So, \bS \in B) + 1}{\log\binom{N-|\St|}{K-|\St|}}.
  \end{equation}

  Note that we can decompose $I(\X_\SSt ; \Y | \X_\St, \So)$ using the following chain of equalities:
  \begin{align*}
    I & (\X_\SSt ; \Y | \X_\St, \So) + I(\bS ; \X_\SSt | \X_\St, \Y, \So) = I(\X_\SSt ; \Y, \bS | \X_\St, \So)  \\
    & = I(\X_\SSt ; \bS | \X_\St, \So) + I(\X_\SSt;  \Y | \X_\St, \bS, \So)= T I(X_\SSt; Y | X_\St, \bS, \So),
  \end{align*}
  where the last equality follows from the independence of $X$ and $\bS$, and the independence of the $(\X, \Y)$ pairs over $t$ given $\bS$. Therefore we have
  \begin{equation} \label{eq:IT}
    I(\X_\SSt ; \Y | \X_\St, \So) = T I(X_\SSt;Y |  X_\St, \bS, \So) - I(\bS ; \X_\SSt | \X_\St, \Y, \So).
  \end{equation}

  From the equality above, we now have that $I(\X_{\SSt}; \Y | \X_{\St}, \So, \bS \in B) \leq T I(X_\SSt; Y |  X_\St, \So, \bS \in B) = T I_\St(B)$, and using this inequality, we obtain the lemma from \eqref{eq:Peb_lb} since $P(E) \geq P(B) P(E | B)$.
\end{proof}

Using Lemma \ref{lemma:beta_lb}, we can readily prove Theorem \ref{thm:eps_lb}.

\subsubsection*{Proof of Theorem \ref{thm:eps_lb}}
  Let \eqref{eq:eps_lb} hold, i.e., $T < (1 - \varepsilon) \frac{\log\binom{N-|\St|}{K-|\St|}}{I_{\St,\varepsilon}}$ for some $\St \subset \So$. Then, there exists a $\gamma > I_{\St,\varepsilon}$ and corresponding set $B_\gamma = \{b \in {\cal B}^K: I_\St(b) < \gamma\}$ such that $P(B_\gamma) > \varepsilon$ and $T < (1-\varepsilon) \frac{\log\binom{N-|\St|}{K-|\St|}}{\gamma}$. From Lemma \ref{lemma:beta_lb}, we have that $P(E) \geq P(B_\gamma) \left( 1 - \frac{T I_\St(B_\gamma) + 1}{\log\binom{N-|\St|}{K-|\St|}} \right)$. Since $I_\St(B_\gamma) \leq \gamma$ by definition, we have 
  \[ 1 - \frac{T I_\St(B_\gamma) + 1}{\log\binom{N-|\St|}{K-|\St|}} > 1 - (1 - \varepsilon) \frac{I_\St(B_\gamma)}{\gamma} - \frac{1}{\log\binom{N-|\St|}{K-|\St|}} \geq \varepsilon - \frac{1}{\log\binom{N-|\St|}{K-|\St|}}. \]
  Since we have $P(B_\gamma) > \varepsilon$, we conclude that $P(E) > \varepsilon \left(\varepsilon - \frac{1}{\log\binom{N-|\St|}{K-|\St|}} \right) \geq \varepsilon^2 - \frac{\varepsilon}{\log(N-K+1)} \geq \frac{\varepsilon^2}{2}$ if \eqref{eq:eps_lb} is true. \hfill \qed

\begin{remark}
It follows from choosing $B = {\cal B}^K$ in Lemma \ref{lemma:beta_lb} that
\begin{equation}\label{eq:avg_lb}
  T \geq \max_{\St \subset S} \frac{\log\binom{N-|\St|}{K-|\St|}}{I(X_\SSt; Y | X_\St, \bS, S)},
\end{equation}
is also a necessary condition for recovery, which involves the average mutual information over the whole space of $\bS$. While this bound is intuitive and may be easier to analyze than the previous lower bounds, it is much weaker than the necessity bound in Theorem \ref{thm:eps_lb} and does not match the worst-case upper bounds we derive in Section \ref{sec:sufficient}.
\end{remark}

\section{Sufficient Conditions for Recovery}
\label{sec:sufficient}

In this section, we derive upper bounds on the number of samples required to recover $S$. We consider models with non-scaling distributions, i.e., models where the observation model, the variable distributions/densities, and the number of relevant variables $|S| = K$ do not depend on scaling variables $N$ or $T$. Group testing as set up in Section \ref{subsec:group_testing} for fixed $K$ is an example of such model. We defer the discussion of models with scaling distributions and $K$ to Section \ref{subsec:scaling}.

To derive the sufficiency bound for the required number of samples, we analyze the error probability of a Maximum Likelihood (ML) decoder \cite{gallager}. For this analysis, we assume that $S_1$ is the true set $\So$ among $\omega \in {\cal I}$. We can assume this w.l.o.g.\ due to the equiprobable support, IID variables and the observation model symmetry assumptions (A1)--(A5). Thus, we can write
\[ P(E) = \frac{1}{\binom{N}{K}} \sum_{\omega \in {\cal I}} \Pr[\hat{S}(\X,\Y)\ne \So | \So] = P(E | S_1). \]
For this reason, we omit the conditioning on $S_1$ in the error probability expressions throughout this section.

The ML decoder goes through all $\binom{N}{K}$ possible sets $\omega \in {\cal I}$ and chooses the set $S_{\omega^*}$ such that
\begin{equation}\label{eq:ML}
  p_{\omega^*}(\Y|\X_{S_{\omega^*}}) > p_\omega(\Y|\X_\So), \quad \forall \omega\ne \omega^*,
\end{equation}
and consequently, an error occurs if any set other than the true set $S_1$ is more likely. This decoder is a minimum probability of error decoder for equiprobable sets, as we assumed in (A1). Note that the ML decoder requires the knowledge of the observation model $p(Y|X_S,\bS)$ and the distribution $p(\bS)$. 

Next, we state our main result. The following theorem provides a sufficient condition on the number of samples $T$ for recovery with average error probability less than $\varepsilon$. 

\begin{theorem} (Sufficiency).
\label{thm:latent_ub}
  For any $0 \leq \varepsilon \leq 1$ and an arbitrary constant $\epsilon > 0$, if
  \begin{equation} \label{eq:latent_ub}
    T > (1+\epsilon) \cdot \max_{\St \subset S} \frac{\log\binom{N-K}{K-|\St|}}{I_{\St,\varepsilon}},
  \end{equation}
  then $\lim_{K \to \infty} \lim_{N \to \infty} P(E) \leq \varepsilon$.
\end{theorem}

\begin{remark}\label{remark:zero_ub}
  Similar to Remark \ref{remark:zero_lb}, we can obtain a sufficient condition for \emph{zero}-error recovery in the limit by considering any sequence $\varepsilon_N \to 0$. In particular, letting $\varepsilon_N = 0$ we have that if 
  \[ T > (1+\epsilon) \cdot \max_{\St \subset S} \frac{\log\binom{N-K}{K-|\St|}}{I_{\St,0}}, \]
  then $\lim_{K \to \infty} \lim_{N \to \infty} P(E) = 0$, matching the necessary condition for recovery in Remark \ref{remark:zero_lb} up to an arbitrarily small constant factor.
\end{remark}

In order to prove Theorem \ref{thm:latent_ub}, we analyze the probability of error in recovery conditioned on $\bS$ taking values in a set $B \subseteq {\cal B}^K$. We state the following lemma, similar in nature to Lemma \ref{lemma:beta_lb} that we used to prove Theorem \ref{thm:eps_lb}.

\begin{lemma}\label{lemma:beta_ub}
  Define the worst-case mutual information constrained to $\bS \in B$ as $\underline{I_\St}(B) = \inf_{b \in B} I_\St(b)$. For any measurable subset $B \subseteq {\cal B}^K$ such that $\bS$ conditioned on $\bS \in B$ is still permutation invariant, a sufficient condition for the error probability conditioned on $\bS \in B$, denoted by $P(E | B)$, to approach zero asymptotically is given by
  \begin{equation}\label{eq:beta_ub}
    T > (1+\epsilon) \cdot \max_{\St \subset S} \frac{\log\binom{N-K}{K-|\St|}}{\underline{I_\St}(B)}.
  \end{equation}
\end{lemma}

Note that the definition of $\underline{I_\St}(B)$ above is different from the definition of $I_\St(B)$ used in Lemma \ref{lemma:beta_lb}, in that one is worst-case while the other is averaged over $\bS \in B$. However, it is noteworthy that the bounds we obtain in Theorems \ref{thm:eps_lb} and \ref{thm:latent_ub} are both characterized by $I_{\St,\varepsilon}$, i.e.\ both are worst-case w.r.t.\ $\bS$ for arbitrarily small $\varepsilon$.

To obtain the sufficient condition in Lemma \ref{lemma:beta_ub}, the analysis starts with a simple upper bound on the error probability $P(E)$ of the ML decoder averaged over all data realizations and observations. 
We define the error event $E_i$ as the event of mistaking the true set for a set which differs from the true set $S_1$ in exactly $i$ variables, thus we can write 
\begin{align}\label{eq:P(Ei)_def}
  P(E_i) = \Pr\left[\exists \omega \ne 1: p_\omega(\Y|\X_\So)\geq p_1(\Y|\X_{S_1}), |\So \setminus S_1|=|S_1 \setminus \So|=i,~|S_1|=|\So|=K\right].
\end{align} 
Using the union bound, the probability of error $P(E)$ can then be upper bounded by
\begin{equation}
  P(E) \leq \sum_{i=1}^K P(E_i) \leq K \max_{i=1,\ldots,K} P(E_i).
\end{equation}

For the proof of Lemma \ref{lemma:beta_ub}, we utilize an upper bound on the error probability $P(E_i|B)$ for each $i = 1,\ldots,K$, corresponding to subsets with $|\St| = K-i$ and conditioned on the event $\bS \in B$. For this subset $B$ of ${\cal B}^K$, the ML decoder constrained to the case $\bS \in B$ is considered and analyzed, such that in its definition in \eqref{eq:ML} the likelihood terms are averaged over $p(\bS | \bS \in B)$. 
This upper bound is characterized by the error exponent $E_o(\rho,b)$ for $\rho \in [0,1]$ and $b \in B$, which is described by
\begin{equation}\label{eq:Eo_b}
  E_o(\rho, b) = -\log \left( \sum_Y \sum_{X_\St} \left[ \sum_{X_\SSt} Q(X_\SSt) p(Y, X_\St | X_\SSt, b)^{\frac{1}{1+\rho}} \right]^{1+\rho} \right).
\end{equation}
For continuous models, sums are replaced with the appropriate integrals and $(\SSt, \St)$ is any partitioning of $S$ to $i$ and $K-i$ variables. 

We present the upper bound on $P(E_i|B)$ in the following lemma. 

\begin{lemma}\label{lemma:P(Ei)}
  The probability of the error event $E_i$ defined in \eqref{eq:P(Ei)_def} that a set selected by the ML decoder differs from the set $S_1$ in exactly $i$ variables conditioned on $\bS \in B$ is bounded from above by
  \begin{align}\label{eq:P(Ei)}
    P(E_i|B) \leq e^{-\left(T E_o(\rho) - \rho\log\binom{N-K}{i} - \log\binom{K}{i}\right)},
  \end{align}
  where we define
  \begin{equation}\label{eq:Eo_lb}
    E_o(\rho) = \inf_{b \in B} E_o(\rho, b) - \frac{\rho}{T} H_{\frac{1}{1+\rho}}(\bS | B),
  \end{equation}
  and $H_{\frac{1}{1+\rho}}(\bS | \bS \in B)$ is the R\'enyi entropy of order $\frac{1}{1+\rho}$ computed for the distribution $p(\bS | \bS \in B)$.\footnote{We refer the reader to Section \ref{sec:setup} for our notation for conditional entropic quantities.}
\end{lemma}
The proof for Lemma \ref{lemma:P(Ei)} is similar in nature to the proof of Lemma III.1 of \cite{group_testing} for discrete variables and observations. It considers the ML decoder defined in \eqref{eq:ML} where the likelihood terms are averaged over $p(\bS | B)$. We note certain differences in the proof and the result, and further extend it to continuous variables and observations in Section \ref{subsec:continuous}. 

We now prove Lemma \ref{lemma:beta_ub}. It follows from a Taylor series analysis of the error exponent $E_o(\rho)$ around $\rho = 0$, from which the worst-case mutual information condition over $\bS \in B$ is derived. This Taylor series analysis is similar to the analysis of the ML decoder in \cite{gallager}.

\subsubsection*{Proof of Lemma \ref{lemma:beta_ub}}

We use the pair of sets $(\Sone, \Stwo)$ for a partition of $S$ to $i$ and $K-i$ variables respectively, instead of $\SSt$ and $\St$. We note that 
\[ \max_{\St \subset S} \frac{\log\binom{N-K}{K-|\St|}}{\underline{I_\St}(B)} = \max_{\substack{i=1,\ldots,K \\ \St \subset S: |\St| = K-i}} \frac{\log\binom{N-K}{i}}{\inf_{b \in B} I(X_\SSt ; Y | X_\St, \bS=b, S)} = \max_{i=1,\ldots,K} \frac{\log\binom{N-K}{i}}{\inf_{b \in B} I(X_\Sone ; Y | X_\Stwo, \bS=b, S)}, \]
for any such partition $(\Sone, \Stwo)$, since the distributions $p(Y | X_S)$, $Q(X)$ and $p(\bS | \bS \in B)$ are all permutation invariant and independent of $S$ except through $X_S$ and $\bS$ as in assumption (A2).\footnote{The actual mutual information expressions we are computing are conditioned on $S = S_1$ (e.g.\ $I(X_\Sone ; Y | X_\Stwo, \bS=b, S = S_1)$) since we assumed the true set is $S_1$ w.l.o.g., which are equal to the averaged mutual information expressions over $S$ (e.g.\ $I(X_\Sone ; Y | X_\Stwo, \bS=b, S)$) which we use throughout this section (see footnote\footnotemark[\getrefnumber{fn:thm}]).}

We derive an upper bound on $P(E|B)$ by upper bounding the maximum probability of the $K$ error events $E_i$, $i=1,\ldots,K$. Using the union bound we have
\begin{equation}
  P(E|B) \leq \sum_{i=1}^K P(E_i|B) \leq K \max_i P(E_i|B) = \max_i \, K P(E_i|B).  \label{eq:Pe_K}
\end{equation}

For each error event $E_i$, we aim to derive a sufficient condition on $T$ such that $K P(E_i|B) \to 0$ as $N \to \infty$, with $P(E_i|B)$ given by \eqref{eq:P(Ei)_def} with additional conditioning on the event $\bS \in B$. Using Lemma \ref{lemma:P(Ei)}, it suffices to find a condition on $T$ such that 
\begin{align}
  T E_o(\rho) - \rho\log\binom{N-K}{i} - \log\binom{K}{i} - \log K \rightarrow\infty,
\label{eq:f_rho_condition}
\end{align}
where $E_o(\rho)$ is given by \eqref{eq:Eo_lb}. 
Note that, since $\log\binom{K}{i} + \log K = \Theta(1)$ for fixed $K$ and $T \to \infty$, the following is a sufficient condition on $T$ for \eqref{eq:f_rho_condition} to hold:
\[ T f(\rho) = T \left(E_o(\rho) - \rho\frac{\log\binom{N-K}{i}}{T} \right) \rightarrow\infty. \]

We note that $E_o(\rho, \bS)$ in \eqref{eq:Eo_lb} does not scale with $N$ or $T$ for non-scaling models. 
To show that the condition \eqref{eq:beta_ub} is sufficient to ensure \eqref{eq:f_rho_condition}, we define $f(\rho) = E_o(\rho) - \rho\frac{\log\binom{N-K}{i}}{T}$ and analyze $E_o(\rho)$ using its Taylor expansion around $\rho = 0$.
Using the mean value theorem, we can write $E_o(\rho,b)$ in the Lagrange form of the Taylor series expansion, i.e., in terms of its first derivative evaluated at zero and a remainder term, 
\begin{equation*}
  E_o(\rho,b) = E_o(0,b) + \rho E_o'(0,b) + \frac{\rho^2}{2} E_o''(\psi,b)
\label{eq:lagrange}
\end{equation*}
for some $\psi \in [0,\rho]$. Note that $E_o(0,b) = 0$ and the derivative of $E_o(\rho,b)$ for any $b \in {\cal B}^K$ evaluated at zero can be shown to be $I_\Stwo(b)$, which is proven in detail in Section \ref{subsec:MI_derivative} in the appendix.

Let $I_\Stwo(b) = I(X_\Sone ; Y | X_\Stwo, \bS=b, S)$ and $\underline{I_\Stwo}(B) = \inf_{b \in B} I_\Stwo(b)$ as defined before. Then, with the Taylor expansion of $E_o(\rho,b)$ above we have
\begin{equation}
  T f(\rho) \geq T \left( \inf_{b} \left[ \rho I_\Stwo(b) + \frac{\rho^2}{2} E_o''(\psi,b) \right] - \rho \frac{H_{\frac{1}{1+\rho}}(\bS | \bS \in B)}{T} - \rho \frac{\log \binom{N-K}{i}}{T} \right) \label{eq:Pe_K2}
\end{equation}
and our aim is to show that the above quantity approaches infinity for some $\rho \in [0, 1]$ as $N \to \infty$.

Now assume that $T$ satisfies
\begin{equation}\label{eq:T_i}
  T > (1+\epsilon) \cdot \frac{\log\binom{N-K}{i}}{\underline{I_\Stwo}(B)}
\end{equation}
for all $i$, which is implied by condition \eqref{eq:beta_ub}. Using the $T$ above and \eqref{eq:Pe_K2} we can then write 
\begin{align*}
  T f(\rho) & \geq T \left( \rho \underline{I_\Stwo}(B) + \frac{\rho^2}{2} \inf_{b} E_o''(\psi,b) - \rho \frac{H_{\frac{1}{1+\rho}}(\bS | \bS \in B)}{T} - \rho \frac{\log \binom{N-K}{i}}{T} \right) \\ 
  & \geq T \left( \rho \underline{I_\Stwo}(B) + \frac{\rho^2}{2} \inf_{b} E_o''(\psi,b) - \rho o(1) - \rho \frac{\underline{I_\Stwo}(B)}{1+\epsilon} \right) \\
  & = T \rho \left( \epsilon' \underline{I_\Stwo}(B) + \frac{\rho}{2} \inf_{b} E_o''(\psi,b) - o(1) \right),
\end{align*}
for $\epsilon' = \frac{\epsilon}{1+\epsilon}$, where in the first inequality we obtain a lower bound by separating the minimum of the sum to the sum of minimums and replacing $T$ in the second inequality, noting that $O(K)/T \to 0$ since $H_{\frac{1}{1+\rho}}(\bS | \bS \in B) = O(K)$. This is due to the inequality $H_{\frac{1}{1+\rho}}(\bS | \bS \in B) \leq H_{\frac{1}{2}}(\bS)$ and the assumption that $H_{\frac{1}{2}}(\bS) = O(K)$.

We note that $E_o''(\psi,b)$ is independent of $N$ or $T$ and thus has bounded magnitude for any $b \in {\cal B}^K$. Then, we pick $\rho$ small enough such that the second derivative term is dominated by the mutual information term; specifically we choose $\rho \leq \frac{\epsilon' I_\Stwo(B)}{|\inf E_o''(\psi,b)|}$ and note that it can be chosen such that $\rho \geq \delta > 0$ for a constant $\delta$, since $|E_o''(\psi,b)| = O(1)$. We then have
\begin{align*}
  T f(\rho) & \geq T \rho \left( \underline{I_\Stwo}(B) [\epsilon' - \epsilon'/2] - o(1) \right) = T \rho \underline{I_\Stwo}(B) \Theta(1) = \log\binom{N-K}{i} \Theta(1) = \Omega(\log N) \to \infty,
\end{align*}
showing that $K P(E_i|B)$ goes to zero for all $i$ given the conditions (A1)-(A5) are satisfied. It follows that $P(E|B) \leq \max_i K P(E_i|B)$ goes to zero for $N \to \infty$ for any $K$, therefore $\lim_{K \to \infty} \lim_{N \to \infty} P(E|B) = 0$. \hfill \qed

Using Lemma \ref{lemma:beta_ub}, we now prove Theorem \ref{thm:latent_ub}.

\subsubsection*{Proof of Theorem \ref{thm:latent_ub}}

First, note that with the definition of $I_{\St,\varepsilon}$, there exists an $\alpha \geq 0$ such that $I_{\St,\varepsilon} - \kappa \leq \alpha \leq I_{\St,\varepsilon}$ for any $\kappa > 0$ (arbitrarily small) and $P(B_\alpha) \leq \varepsilon$ where $B_\alpha = \{ b \in {\cal B}^K: I_\Stwo(b) < \alpha \}$. Note that this subset preserves the permutation invariance property of $\bS$ and $\underline{I_\St}(B_\alpha^c) \geq \alpha \geq I_{\St,\varepsilon} - \kappa$.

We then have using Lemma \ref{lemma:beta_ub} that if 
\begin{equation}\label{eq:kappa}
  T > (1+\epsilon) \max_{\St \subset S} \frac{\log\binom{N-K}{K-|\St|}}{I_{\St,\varepsilon} - \kappa},
\end{equation}
then there exists an ML decoder such that $P(E | B_\alpha^c) \to 0$. However since $\kappa$ is arbitrarily small, \eqref{eq:kappa} is satisfied for any $T$ satisfying condition \eqref{eq:beta_ub} in Lemma \ref{lemma:beta_ub}, i.e., that $T > (1+\epsilon) \max_{\St \subset S} \frac{\log\binom{N-K}{K-|\St|}}{I_{\St,\varepsilon}}$.

We can write $P(E) = P(B_\alpha) P(E | B_\alpha) + P(B_\alpha^c) P(E | B_\alpha^c) \leq P(B_\alpha) + P(E | B_\alpha^c)$ and we have $P(E | B_\alpha^c) \to 0$ and $P(B_\alpha) \leq \varepsilon$. Therefore, we have shown that $P(E) \leq \varepsilon$ is achievable (specifically with the ML decoder that considers $\bS \in B_\alpha^c$) and the theorem follows. \hfill \qed

\subsection{Sufficiency for Models with Scaling}
\label{subsec:scaling}

In this section, we consider models with scaling distributions, i.e., models where the observation model and the variable distributions/densities and number of relevant variables $|S| = K$ may depend on scaling variables $N$ or $T$. 
While Theorem \ref{thm:latent_ub} characterizes precisely the constants (including constants related to $K$) in the sample complexity, it is also important to analyze models where $K$ is scaling with $N$ or where the distributions depend on scaling variables. Group testing where $K$ scales with $N$ (e.g.\ $K = \Theta(\sqrt{N})$) is an example of such model, as well as the normalized sparse linear regression model when the SNR and the random matrix probabilities are functions of $N$ and $T$ in Section \ref{subsec:lcs}. Therefore, in this section we consider the most general case where $Q(X_n)$ or $p(Y | X_S)$ can be functions of $K$, $N$ or $T$ and $K = O(N)$. Note that the necessity result in Theorem \ref{thm:eps_lb} also holds for scaling models thus does not need to be generalized. As we noted in the problem setup, we consider $N$ as the independent scaling variable and $K = K(N)$, $T = T(N)$ scale as functions of $N$.

To extend the results in Section \ref{sec:sufficient} to general scaling models, we employ additional technical assumptions related to the smoothness of the error exponent \eqref{eq:Eo_b} and its dependence on the latent observation model parameter $\bS$. For a proper subset $\St \subset S$, we consider the error exponent $E_N(\rho, \bS) = E_o(\rho, \bS)$ as defined in \eqref{eq:Eo_b}, which we subscript with $N$ to emphasize its dependence on the scaling variable $N$. 
Throughout this section we also modify our notation w.r.t.\ $\bS$, assuming the existence of a ``sufficient statistic'' $s = T(\bS)$ that will be formalized shortly. Instead of writing quantities as functions of $b \in {\cal B}^K$, such as $I_\St(b)$ defined in \eqref{def:cond_MI} and $E_N(\rho,b)$ as defined in \eqref{eq:Eo_b}, we write them as functions of $s$, e.g.\ $I_\St(s)$ and $E_N(\rho,s)$. 
We also define the following quantity that will be utilized in our smoothness conditions.

\begin{definition}\label{def:F_N}
  For a proper subset $\St \subset S$, the normalized first derivative of the error exponent is $F_N(\rho,s) = \frac{\frac{\partial}{\partial \rho} E_N(\rho,s)}{I_\St(s)}$.
\end{definition}

With this definition, we formally enumerate below the regularity conditions we necessitate for scaling models for each $\St \subset S$.

\begin{enumerate}[({RC}1)]
  \item There exists a sufficient statistic $s = T(b)$ for $b \in {\cal B}^K$ such that the error exponent $E_N(\rho,b)$ only depends on $T(b)$, i.e.\ we have $E_N(\rho, b) = E_N(\rho, b')$ for all $b$ and $b'$ that satisfy $T(b) = T(b')$. In addition, $s$ belongs to a compact set ${\cal C} = \bigcup_K T({\cal B}^K) \subset \mathbb{R}^d$ for a constant $d$ independent of $N$.
  \item $F_\infty(\rho,s) = \lim_{N \to \infty} F_N(\rho,s)$ exists for each $\rho$ and $s$, it is continuous in $\rho$ for each $s$, and the convergence is uniform in $s$. 
\end{enumerate}

We note that the first condition is trivially satisfied when $K$ is fixed by letting ${\cal C} = {\cal B}^K$, or when $\bS$ is fixed or does not exist by letting ${\cal C}$ be a singleton set, e.g.\ in group testing. In other cases, a sufficient statistic is frequently the average power of (parts of) the vector $\bS$, e.g.\ in sparse linear regression, which we consider in Section \ref{subsec:lcs}.

The reason we consider the quantity $F_N(\rho,s)$ as defined in Definition \eqref{def:F_N} is that it is normalized such that $F_N(\rho,s) \in [0,1]$ for all $\rho$ and $s$, $F_N(0,s) = 1$ for all $s$ and non-increasing in $\rho$ since $E_N(\cdot,s)$ is a concave function. Indeed, it is not possible to directly consider the limit of the first derivative $E'_N(\rho,s)$ or the mutual information $I_\St(s)$, since in most applications these quantities either converge to the zero function or diverge to infinity as $N$ increases.

Given the regularity conditions (RC1) and (RC2) are satisfied, we have the following sufficient condition analogous to Theorem \ref{thm:latent_ub}.

\begin{theorem}\label{thm:scaling_ub}
  For any $0 \leq \varepsilon \leq 1$ and a constant $C$, if 
  \begin{equation}\label{eq:scaling_ub}
    T > C \cdot \max_{\St \subset S} \frac{\log\binom{N-K}{K-|\St|} + H_\frac{1}{2}(\bS)}{I_{\St,\varepsilon}},
  \end{equation}
  then $\lim_{N \to \infty} P(E) \leq \varepsilon$.
\end{theorem}

We also note the extra R\'enyi entropy term $H_\frac{1}{2}(\bS)$ in the numerator compared to Theorem \ref{thm:latent_ub}. This term results from the uncertainty present in the random $\bS$, whose value is unknown to the decoder. The bound reduces to the non-scaling bound for fixed $K$ or certain scaling regimes of $K$ since this term will be dominated by the other term in the numerator. It also disappears asymptotically when partial recovery is considered such that we maximize over $|\St| \leq \alpha K$ for a constant $\alpha$. The necessity for a constant factor $C$ stems from the $\log\binom{K}{i}$ term in the error exponent and the $\log K$ term due to union bounding over $i = 1,\ldots,K$ as in the proof of Lemma \ref{lemma:beta_ub}.

\subsubsection*{Proof of Theorem \ref{thm:scaling_ub}}

We prove the analogue of Lemma \ref{lemma:beta_ub}, omitting the dependencies on $B$ for brevity. The rest of the proof follows from the same arguments used for proving Theorem \ref{thm:latent_ub} given Lemma \ref{lemma:beta_ub}. 

Similar to the proof of Lemma \ref{lemma:beta_ub}, we want to show that $K P(E_i) \to 0$ as $N \to \infty$ for any $i = 1,\ldots,K$. For each $i$, we consider an arbitrary partition $(\Sone, \Stwo)$ of $S$ to $i$ and $K-i$ elements such that $\St = \Stwo$. To show that $K P(E_i) \to 0$, it suffices to show that there exists a $\rho > 0$ for which
\[ T f_N(\rho) = T E_N(\rho) - \rho \log\binom{N-K}{i} - \log\binom{K}{i} - \log K \to \infty. \]
Using the inequality $\log\binom{K}{i} + \log K \leq 2\log\binom{N-K}{i}$, we can lower bound $T f_N(\rho)$ as
\[ T f_N(\rho) \geq T \inf_s E_N(\rho,s) - (\rho+2) D_N, \]
where we define $D_N \triangleq \log\binom{N-K}{i} + H_{\frac{1}{2}}(\bS)$.

Define the ratio $R_N(\rho,s) = \frac{E_N(\rho,s)}{\rho I_\Stwo(s)}$ and note that the derivative of $\bar{R}_N(\rho,s) \triangleq \rho R_N(\rho,s) = \frac{E_N(\rho,s)}{I_\Stwo(s)}$ w.r.t.\ $\rho$ is $F_N(\rho,s)$. Since $E_N(0,s) = 0$ therefore $\bar{R}_N(0,s) = 0$, from the Lagrange form of the Taylor expansion of $\bar{R}_N(\rho,s)$ around $\rho = 0$ we have $R_N(\rho,s) = F_N(\psi_s,s)$ for some $\psi_s \in [0, \rho]$. Noting that $E_N(\rho,s)$ is concave \cite{gallager} therefore $F_N(\rho,s)$ is non-increasing in $\rho$, it follows that $R_N(\rho,s) \geq F_N(\rho,s)$ for any $\rho$ and $s$.

We present the following technical lemma, which is proved in the appendix.

\begin{lemma}\label{lemma:F_N_bd}
  For any $c > 0$, there exists a constant $\rho_c > 0$ and integer $N_0$ such that for all $N \geq N_0$, $F_N(\rho_c, s) \geq 1-c$, for all $s \in {\cal C}$.
\end{lemma}

Let $0 < c < 1$ be a constant. From Lemma \ref{lemma:F_N_bd}, there exists $\rho_c > 0$ and $N_0$ such that for all $N \geq N_0$, we have that $F_N(\rho_c,s) \geq 1-c$ for all $s$. Using the bound on $T$ in \eqref{eq:scaling_ub}, we then have the chain of inequalities
\[ \frac{T \inf_s E_N(\rho_c,s)}{\rho_c D_N} > \frac{C \inf_s E_N(\rho_c,S)}{\rho_c \inf_s I_\Stwo(s)} \geq C \inf_s \frac{E_N(\rho_c,s)}{\rho_c I_\Stwo(s)} = C \inf_s R_N(\rho_c,s) \geq C (1-c). \]

We then have $T \inf_s E_N(\rho_c,s) > C (1-c) \rho_c D_N$ for constants $c$ and $\rho_c$. Therefore,
\[ T f_N(\rho_c) \geq T \inf_s E_N(\rho_c,s) - (\rho_c+2) D_N > \left(C (1-c) \rho_c - \rho_c - 2\right) D_N \geq c' D_N \to \infty, \]
for any constant $C$ such that $C \geq \frac{c' + 2}{\rho_c (1-c)}$, proving that $K P(E_i) \to 0$ for each $i$ and therefore $P(E) \to 0$. \hfill \qed

While the resulting upper bound \eqref{eq:scaling_ub} in Theorem \ref{thm:scaling_ub} has the same mutual information expression in the denominator as Theorem \ref{thm:latent_ub}, it has an extra $H_{\frac{1}{2}}(\bS)$ term in the numerator in addition to the combinatorial term $\log\binom{N-K}{K-|\St|}$. While this term is negligible for high sparsity regimes, it might affect the sample complexity in regimes where $K$ does not scale too slowly and $\bS$ has uncorrelated (e.g.\ IID) elements such that the entropy of $\bS$ is high. In such cases, the $H_{\frac{1}{2}}(\bS)$ term related to the uncertainty in $\bS$ may dominate the combinatorial term $\log\binom{N-K}{K-|\St|}$ for large subsets $\St$. 

To this end, we state the following theorem that uses the results of Theorem \ref{thm:scaling_ub} and establishes guarantees that all but a vanishing fraction of the indices in the support can be recovered reliably. 

\begin{theorem}\label{thm:partial}
  For any $0 \leq \varepsilon \leq 1$ and a constant $C$, if $H_\frac{1}{2}(\bS) = O(K)$ (satisfied by IID $\bS$), $K = O(N/\log N)$ and 
  \begin{equation}\label{eq:partial_ub}
    T > C \cdot \max_{\St \subset S} \frac{\log\binom{N-K}{K-|\St|}}{I_{\St,0}},
  \end{equation}
  then with probability one $\limsup_{N \to \infty} \frac{\left|S \setminus \hat{S}_N(\X,\Y)\right|}{|S|} = 0$.
\end{theorem}

We prove the theorem in the appendix. 
We note that Theorem \ref{thm:partial} does not give guarantees on \emph{exact} recovery of all elements of $S$ as in Theorem \ref{thm:scaling_ub}, however it gives guarantees that a set that overlaps the true set $S$ in all but an arbitrarily small fraction of elements can be recovered. We also considered the special case of recovery with zero error probability, i.e., $\varepsilon = 0$ in this analysis, however it can be extended to the non-zero error probability analysis similar to Theorem \ref{thm:scaling_ub}.

\subsection{Discussion}
\label{sec:discussion}

\noindent
\textbf{Tight characterization of sample complexity: upper vs.\ lower bounds.}
We have shown and remarked in Sections \ref{sec:necessary} and \ref{sec:sufficient} that for arbitrarily small recovery error $\varepsilon \to 0$, the upper bound for non-scaling models in Theorem \ref{thm:latent_ub} is tight as it matches the lower bound given in Theorem \ref{thm:eps_lb}. The upper bound in Theorem \ref{thm:scaling_ub} for arbitrary models is also tight up to a constant factor $C$ provided that mild regularity conditions on the problem hold.

\noindent
\textbf{Partial recovery.} 
As we analyze the error probability separately for $i=1,\ldots,K$ support errors corresponding to $\St \subset S$ with $|S \setminus \St| = i$ in order to obtain the necessity and sufficiency results, it is straightforward to determine necessary and sufficient conditions for \emph{partial} support recovery instead of \emph{exact} support recovery. By changing the maximization from over all subsets $\St \subset S$ (i.e.\ $i = 1,\ldots,K$) to $\St \subset S$ such that $|\St| < k$ (i.e.\ $i = K-k+1,\ldots,K$) in the recovery bounds, the conditions to recover at least $k$ of the $K$ support indices can be determined.

\noindent
\textbf{Technical issues with typicality decoding.} 
It is worth mentioning that a typicality decoder can also be analyzed to obtain a sufficient condition, as used in the early versions of \cite{arxiv_v2}. However, typicality conditions must be defined carefully to obtain a tight bound w.r.t.\ $K$, as with standard typicality definitions the atypicality probability may dominate the decoding error probability in the typical set. For instance, for the group testing scenario considered in \cite{group_testing}, where $X_n \sim \text{Bernoulli}(1/K)$, we have $\Pr[X_S = (1,\ldots,1)] = (1/K)^K$, which would require the undesirable scaling of $T$ as $K^K$, to ensure typicality in the strong sense (as needed to apply results such as the packing lemma \cite{abbas}). Redefining the typical set as in \cite{arxiv_v2} is then necessary, but it is problem-specific and makes the analysis cumbersome compared to the ML decoder adopted herein and in \cite{group_testing}. Furthermore, the case where $K$ scales together with $N$ requires an even more subtle analysis, whereas the analysis of the ML decoder analysis is more straightforward in regards to that scaling. Typicality decoding has also been reported as infeasible for the analysis of similar problems, such as multiple access channels where the number of users scale with the coding block length \cite{guo}.

\section{Applications with Linear Observations}
\label{sec:apps_linear}

In this section and the next section, we establish results for several problems for which our necessity and sufficiency results are applicable. For this section we focus on problems with linear observation models and derive results for sparse linear regression, considering several different setups. Then, we consider a multivariate regression model, where we deal with vector-valued variables and outcomes.

\subsection{Sparse Linear Regression}
\label{subsec:lcs}

Using the bounds presented in this paper for general sparse models, we derive sufficient and necessary conditions for the sparse linear regression problem with measurement noise \cite{donoho} and a Gaussian variable matrix with IID entries. 

We consider the following model similar to \cite{shuchin},
\begin{equation}
  \Y = \X \beta + \bm{W},
\end{equation}
where $\X$ is the $T \times N$ variable matrix, $\beta$ is a $K$-sparse vector of length $N$ with support $S$, $\bm{W}$ is the measurement noise of length $T$ and $\Y$ is the observation vector of length $T$. In particular, we assume $X_n^{(t)}$ are Gaussian distributed random variables and the entries of the matrix are independent across rows $t$ and columns $n$. Each element $X_n^{(t)}$ is zero mean and has variance $\sigma_x^2$. $\bm{W}$ denotes the observation noise of length $T$. We assume each element is IID with $W \sim {\cal N}(0, \sigma_w^2)$. The coefficients of the support, $\bS$, are IID random variables with $\bmin \leq \beta_k^2 \leq \bmax$ and (continuous) R\`enyi entropy $H_\frac{1}{2}(\beta_k) = h$ for $k \in S$. W.l.o.g.\ we assume that $h, \bmin, \bmax$ are constants, since their scaling can be incorporated into $\sigma_x$ or $\sigma_w$ instead.

In order to analyze the sample complexity using Theorems \ref{thm:eps_lb}, \ref{thm:latent_ub} and \ref{thm:scaling_ub}, we need to compute the worst-case mutual information $I_{\St,\varepsilon}$. 
We first compute the mutual information $I_\St(b) = I(X_\SSt ; Y | X_\St, \bS = b, S)$ for $|\SSt| = i$.
\begin{align*}
  I_\St(b) = I(X_\SSt ; Y | X_\St, \bS = b, S) & = h(Y | X_\St, \bS = b, S) - h(Y | X_S, \bS = b, S) \\
  & = h\left( X_\SSt^\top b_\SSt + W | b_\SSt \right) - h(W) \\
  & = \frac{1}{2} \log\left( 2\pi e \left(\text{var}\left(X_\SSt^\top b_\SSt | b_\SSt \right) + \sigma_w^2\right)\right) - \frac{1}{2} \log \left( 2\pi e \; \sigma_w^2\right) \\
  & = \frac{1}{2} \log \left( 1 + \frac{\|b_\SSt\|^2 \sigma_x^2}{\sigma_w^2} \right), 
\end{align*}
where the second equality follows from the independence of $X_\SSt$ and $X_\St$ and the last equality follows from the fact that $\text{var}(X_\SSt^\top b_\SSt | b_\SSt) = b_\SSt^\top E[X_\SSt X_\SSt^\top]b_\SSt  =  b_\SSt^\top b_\SSt \sigma_x^2$.

Assuming there is a non-zero probability that $\beta_k^2$ is arbitrarily close to $\bmin$, it is easy to see that for $\varepsilon = 0$,
\begin{equation*}
  I_{\St,0} = \frac{1}{2} \log \left( 1 + \frac{i \bmin \sigma_x^2}{\sigma_w^2} \right).
\end{equation*}
For this problem it is also possible to compute the exact error exponent $E_N(\rho,b)$, which we do in the analysis in Section \ref{subsec:lcs_analysis} and prove that the regularity conditions (RC1-2) hold for Theorem \ref{thm:scaling_ub} in the appendix. The bounds in the theorem we present below then follow from Theorems \ref{thm:eps_lb}, \ref{thm:latent_ub} and \ref{thm:partial} respectively.

\begin{theorem}\label{thm:lcs}
  For sparse linear regression with the setup described above, a necessary condition on the number of measurements for exact recovery of the support is
  \begin{equation}\label{eq:cs_lb}
    T \geq 2 \max_{i=1,\ldots,K} \frac{\log\binom{N-K+i}{i}}{\log \left( 1 + \frac{i \bmin \sigma_x^2}{\sigma_w^2} \right)},
  \end{equation}
  a sufficient condition for exact recovery for constant $K$ and $\sigma_x, \sigma_w$ independent of $N$ is
  \begin{equation}\label{eq:cs_ub}
    T \geq (2+\epsilon) \max_{i=1,\ldots,K} \frac{\log\binom{N-K}{i}}{\log \left( 1 + \frac{i \bmin \sigma_x^2}{\sigma_w^2} \right)},
  \end{equation}
  for an arbitrary $\epsilon > 0$. A sufficient condition for recovery with a vanishing fraction of support errors for $K=O(N/\log N)$ is
  \begin{equation}\label{eq:cs_ub_scaling}
    T \geq C \max_{i=1,\ldots,K} \frac{\log\binom{N-K}{i}}{\log \left( 1 + \frac{i \bmin \sigma_x^2}{\sigma_w^2} \right)},
  \end{equation}
  for a constant $C$.
\end{theorem}

We now evaluate the above bounds for different setups and compare against standard bounds in the sparse linear regression and compressive sensing literature. We specifically compare against \cite{wang} which presents lower bounds, \cite{Rad-11} that provides upper bounds matching \cite{wang} and \cite{shuchin} which presents lower and upper bounds on measurements and a lower bound on SNR. Note that the setups of \cite{wang,Rad-11} and \cite{shuchin} are different but equivalent for certain cases, however the setup of \cite{shuchin} allows for a unique analysis of the SNR, which is the reason we include it in this section. 

\subsubsection*{Comparison to \cite{wang} and \cite{Rad-11}}
First, we compare against the lower and upper bounds in \cite{wang} and \cite{Rad-11} respectively, presented in Table 1 in \cite{Rad-11}. In this setup, we have $\sigma_x^2 = \sigma_w^2 = 1$ and we will compare for the lower SNR regime $\bmin = \Theta(1/K)$ and the higher SNR regime $\bmin = \Theta(\log K/K)$. 

The lower bounds we state are for the general regime $K = O(N)$, while the upper bounds are for $K = O(N/\log N)$, as we note in Theorem \ref{thm:lcs}. For $\bmin = \Theta(1/K)$, we have that $I_{\St,0} = \Theta(\log(1+i/K)) = \Theta(i/K)$, therefore for both the lower and the upper bounds we have $T = \Omega\left(\max_i \frac{i \log(N/i)}{i/K}\right) = \Omega(K \log N)$, matching \cite{wang,Rad-11} for both sublinear and linear sparsity. 

For $\bmin = \Theta(\log K/K)$, we have upper and lower bounds $T = \Omega\left(\max_i \frac{i \log(N/i)}{\log(1 + i \log K/K)} \right)$. For linear sparsity, let $i = \log K$, for which the numerator is $\Theta(\log K \log(N/\log K)) = \Theta(\log^2 N)$ and denominator $\Theta(\log(1 + \log^2 K/K)) = \Theta(\log^2 N / N)$, thus we can obtain $T = \Omega(N)$ for both lower and upper bounds. For sublinear sparsity, first consider $i = K$. For this $i$, we obtain $T = \Omega\left(\frac{K \log(N/K)}{\log (\log K)}\right)$ directly. Second, considering $i = K/\log K$, we get $T = \Omega\left( \frac{K}{\log K} \log(N \log K / K) \right) = \Omega\left(\frac{K \log N}{\log K}\right)$. Thus we match the upper and lower bounds as the maximum of these two cases. 
Matching bounds can also be shown for the case $\bmin = \Theta(1)$, however we omit the analysis for this case for brevity.

\subsubsection*{Comparison to \cite{shuchin}}

Next, we compare with the bounds for exact recovery derived in \cite{shuchin} where we have $\sigma_x^2 = \frac{1}{T}$, $\bmin = \Theta(1)$ and $\sigma_w^2 = \frac{1}{\SNR}$. For this setup, we prove that $\SNR = \Omega(\log N)$ is a necessary condition for recovery and for that $\SNR$, $T = \Omega(K \log (N/K))$ is necessary for sublinear $K = O(N^p)$, $p < 1$ and sufficient in the same regime for recovery with a vanishing fraction of support errors, which matches the conditions derived in \cite{shuchin} for the corresponding sparsity conditions. 

We remark that $\SNR = \Theta(\log N)$ regime in this model roughly corresponds to the $\bmin = \Theta(1/K)$ regime in \cite{wang,Rad-11}.
We also note that while having $\sigma_x^2$ depend on $T$ complicates the derivation of lower and upper bounds, this scaling ensures normalized columns and conveniently decouples the effects of SNR and the number of measurements. The decoupling leads to the aforementioned lower bound on SNR that is independent of the number of measurements. 

We now provide the analysis to obtain the above conditions given Theorem \ref{thm:lcs}. 
We first show that $\SNR = \Omega(\log N)$ is necessary for recovery. For any $N$, $K$ or SNR assume $T$ scales much faster, e.g.\ $T = \omega(K \bmin \SNR)$, such that
\[ I_{\St,0} = \frac{1}{2} \log \left( 1 + \frac{i \bmin \SNR}{T} \right) \asymp \frac{1}{2} \frac{i \bmin \SNR}{T}, \]
since $\log(1+x) = \Theta(x)$ for $x \to 0$. Then, the necessary condition given by \eqref{eq:cs_lb} is
\[ T > 2 \max_i \frac{\log \binom{N-K+i}{i}}{\frac{i \bmin \SNR}{T}} \]
which readily leads to the condition that
\begin{equation} \label{eq:SNR}
  \SNR > 2 \max_i \frac{\log \binom{N-K+i}{i}}{i \bmin} \asymp \max_i \log (N/i) = \log N 
\end{equation} 
for constant $\bmin$. Note that, if the condition above is necessary for any $T = \omega(K \sigma^2 \SNR)$, it is also necessary for smaller scalings of $T$.

For the lower bound, we consider sublinear sparsity $K=O(N^p)$ and $\SNR = O(\log N)$ and prove that $T = \Omega(K \log (N/K))$ is necessary by contradiction. Let $i=K$ and assume that $T = c_N K \log(N/K)$, where $c_N \to 0$. In the left-hand side of the inequality \eqref{eq:cs_lb}, we have $T = c_N K \log(N/K)$, while on the right-hand side we have
\[ \frac{2 \log\binom{N}{K}}{\log\left(1 + \frac{K \SNR}{c_N K \log(N/K)}\right)} = O\left(\frac{K \log(N/K)}{\log\left(1 + \frac{\alpha}{c_N}\right)}\right), \]
for some constant $\alpha$ noting that $K \log(N/K) = \Theta(K \log N)$. 
Canceling the $K \log(N/K)$ terms on each side we have $c_N \log(1 + \frac{\alpha}{c_N}) = o(1)$ and therefore the inequality is not satisfied for any $c_N \to 0$.

We now show that $T = \Omega(K \log N/K) = \Omega(K \log N)$ is a sufficient condition for $K = O(N^p)$ and $\SNR = \Theta(\log N)$. For $T = \Theta(K \log N)$, the right-hand side in the sufficient condition given in \eqref{eq:cs_ub_scaling} is
\[ \Theta\left(\max_i \frac{i \log (N/i)}{\log \left(1 + \frac{i}{K} \bmin \right)}\right), \]
where we note that $K = O(i \log(N/i))$ for all $i$ in this scaling regime. For $i = o(K)$ above is equivalent to $\Theta\left(\frac{i \log(N/i)}{\frac{i}{K}}\right) = \Theta(K \log N)$. For $i = \Theta(K)$ we have the denominator $\Theta(1)$ therefore the term above is again $\Theta(K \log N)$. Thus $T = \Theta(K \log N)$ satisfies \eqref{eq:cs_ub_scaling}.

In Figure \ref{fig:cs_bound}, we illustrate the lower bound on the number of observations for the setup of \cite{shuchin}, which shows that a necessary condition on $\SNR$ has to be satisfied for recovery, as we have stated above. 

\begin{figure}
  \centering
  \includegraphics[width=0.5\textwidth]{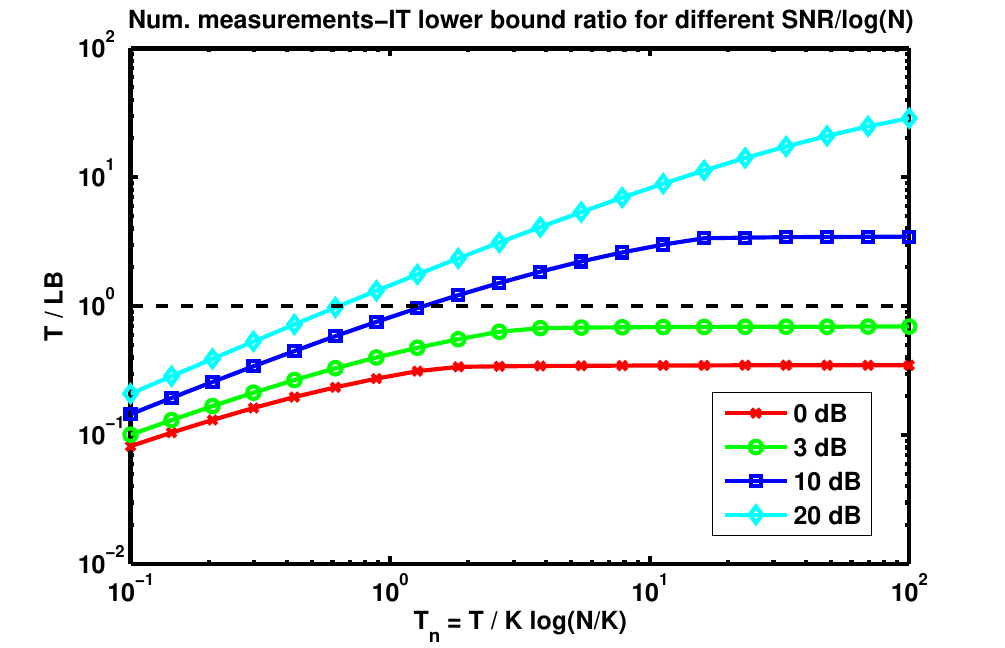}
  \caption{$\frac{T}{LB}$ vs.\ $T$ for different $\SNR$ values, where $LB$ is the necessity bound given by \eqref{eq:cs_lb} for $K = 16$, $D = 512$ and $\bmin = \bmax = 1$. For low levels of $\SNR$ the necessary condition ($T > LB$, above the dotted line) is not satisfied even for very large $T$, for fixed $K$ and $N$. This is due to $\log\left(1 + c \, \frac{\SNR}{T}\right)$ behaving linearly instead of logarithmically for low $\frac{\SNR}{T}$ ratios.}
  \label{fig:cs_bound}
\end{figure}

\begin{remark}
  We showed that our relatively simple mutual information analysis gives us upper and lower bounds that are asymptotically identical to the best-known bounds obtained through problem-specific analyses in \cite{wang,Rad-11,shuchin} in their respective setups for most scaling regimes of interest. 
  
  We also note that while the aforementioned work analyze bounds assuming a lower bound on the power of the support coefficients $\bmin$, our lower bound analysis proves that such a lower bound is \emph{required} for recovery. For instance, assuming there exists any one index $k \in S$ such that $\beta_k = 0$ with non-zero probability, for $\St = S \setminus \{k\}$ we would obtain $I_{\St,0} = 0$, showing that recovery is impossible due to Theorem \ref{thm:eps_lb}.
\end{remark}

Another interesting aspect of our analysis is that in addition to sample complexity bounds, an upper bound to the probability of error in recovery can be explicitly computed and obtained using Lemma \ref{lemma:P(Ei)} for any finite triplet $(T,K,N)$, using the error exponent $E_N(\rho,b)$ obtained in Section \ref{subsec:lcs_analysis}. Following this line of analysis, an upper bound is obtained for sparse linear regression and compared to the empirical performance of practical algorithms such as Lasso \cite{candesplan,wainwright-lasso} in \cite{aistats}. It is then seen that while certain practical recovery algorithms have provably optimal asymptotic sample complexity, there is still a gap between information-theoretically attainable recovery performance and the empirical performance of such algorithms. We refer the reader to \cite{aistats} for details.

\subsection{Regression with Correlated Support Elements}
\label{subsec:correlated}

In the following subsections we consider several variants of the sparse linear regression problem with different setups. First, we consider a variant where the support elements $\bS$ are correlated, in contrast to the IID assumption we had in the previous section. Note that we are still considering recovery in a Bayesian setting for $\bS$ rather than a worst-case analysis. Having correlated elements in the support usually complicates the analysis when using problem-specific approaches, however, in our framework the analysis is no different than the IID case. We even obtain slightly improved bounds (which we detail shortly) as a result of correlation decreasing the uncertainty in the observation model.

Formally, we consider the same problem setup as above, except that $\bS$ is not IID and we assume $H_\frac{1}{2}(\bS) = O(1)$. A special case of such a correlated setup is when the distribution $p(\bS)$ has finite volume on its support, for which we have $H_\frac{1}{2}(\bS) \leq H_0(\bS) = |\mathrm{supp}(p(\bS))| = O(1)$. We note that the correlation in $\bS$ does not affect the mutual information computation for $I_{\St,0}$ nor the analysis to show that the regularity conditions (RC1-2) hold. The only change in the analysis is that the R\'enyi entropy term $H_\frac{1}{2}(\bS)$ in the numerator of Theorem \ref{thm:scaling_ub} is now asymptotically dominated by the combinatorial term $\log\binom{N-K}{K-|\St|}$ for all $\St \subset S$. Thus, we can improve the upper bound for scaling $K$ in \eqref{eq:cs_ub_scaling} from recovery with a vanishing fraction of errors to exact recovery, and from $K = O(N/\log N)$ to $K = O(N)$ to obtain the following theorem as the analogue of Theorem \ref{thm:lcs}.

\begin{theorem}\label{thm:lcs_corr}
  For sparse linear regression with correlated support elements $\bS$, a necessary condition on the number of measurements for exact recovery of the support is
  \begin{equation}\label{eq:cs_lb_corr}
    T \geq 2 \max_{i=1,\ldots,K} \frac{\log\binom{N-K+i}{i}}{\log \left( 1 + \frac{i \bmin \sigma_x^2}{\sigma_w^2} \right)},
  \end{equation}
  and a sufficient condition for exact recovery is
  \begin{equation}\label{eq:cs_ub_corr}
    T \geq C \max_{i=1,\ldots,K} \frac{\log\binom{N-K}{i}}{\log \left( 1 + \frac{i \bmin \sigma_x^2}{\sigma_w^2} \right)},
  \end{equation}
  for a constant $C$.
\end{theorem}

Evaluating our bounds in the setup of \cite{wang,Rad-11}, we observe that the upper and lower bounds are unchanged, however the upper bounds we obtain are for exact recovery instead of recovery with a vanishing fraction of support errors and we are no longer restricted to $K = O(N/\log N)$.

\subsection{Sensing Matrices with Correlated Columns}

We consider another extension to the typical sparse linear regression setup in Section \ref{subsec:lcs}, where we have correlated columns in the sensing matrix $\X$. As with correlated support coefficients, this is yet another setup whose analysis in the classical sparse linear regression literature is inherently more cumbersome than the IID case. While our problem setup does not support non-IID variables $X$, we consider a simple extension that appeared in previous work in \cite{aistats} and show that correlation only affects the effective SNR, and up to a constant amount of correlation can be theoretically tolerated. This result is in contrast to the earlier results concerning the performance of algorithms such as Lasso, which considered decaying correlations \cite{candesplan}. An information-theoretic analysis of this setup has also been considered in \cite{wainwright}.

Formally, we consider the setup of Section \ref{subsec:lcs}, with the difference that for any two elements in $X_j$, $X_k$, $j \neq k$ on a row of $\X$, we have a correlation coefficient $\rho > 0$. For instance for the setup of \cite{shuchin}, this corresponds to $E[X_j X_k] = \frac{\rho}{T}$. We remark that this probabilistic model is equivalent to the following one: Let $X_k = \mu + U_k$, where $\mu \sim {\cal N}(0, \rho \sigma_x^2)$ and $U_k \sim {\cal N}(0, (1-\rho)\sigma_x^2)$ where $U_k$ is IID across $k = 1,\ldots,N$. As a result, we have that $X_k$ for $k=1,\ldots,N$ are conditionally IID given the latent factor $\mu$.

We note that the general framework we consider in Sections \ref{sec:necessary} and \ref{sec:sufficient} explicitly necessitates IID variables $X = (X_1,\ldots,X_N)$. However, this setup can be naturally extended to \emph{conditionally IID} variables $X$ conditioned on a latent factor $\theta$ (see \cite{aistats} for details). It turns out that most of the results from Sections \ref{sec:necessary} and \ref{sec:sufficient} readily generalize to this model, with the dependence having the effect that the mutual information expressions are additionally conditioned on the latent factor $\theta$, i.e., $I_\St(b) = I(X_\SSt ; Y | X_\St, \bS = b, S, \theta)$.

With the extension to conditionally IID variables considered in \cite{aistats} and the formulation of the correlated columns setup as conditionally IID columns given $\mu$, the mutual information expressions and regularity conditions conditioned on $\mu$ can be explicitly computed. With these results, we observe that the correlated columns problem is equivalent to the IID problem except that the ``effective SNR'' $\SNR_e = (1-\rho) \SNR$, where $\SNR$ denotes the signal to noise ratio for the IID problem. This leads to the conclusion that up to a constant correlation can be tolerated for recovery, in contrast to older results on the analysis of recovery algorithms such as Lasso which required decaying correlations, e.g.\ $\rho = O(1/\log N)$ \cite{candesplan}. We refer the reader to \cite{aistats} for extension to the conditionally IID framework, along with analysis and numerical experiments for the correlated sensing columns setup.

\subsection{Bouquet Model for Support Elements}

In this subsection, we consider another variation of the linear regression problem, where for each sample $t$,
\[ Y^{(t)} = \langle X^{(t)}, \beta^{(t)} \rangle + W^{(t)}, \]
with each $\beta^{(t)}$ having the same support $S$, but different coefficients $\bS^{(t)}$ obeying a ``bouquet model''
\[ \bS^{(t)} = \bS^{(0)} + V^{(t)}. \]
We assume $V^{(t)}$ is IID across samples $t$ and is described by a zero-mean Gaussian with variance $\sigma_v^2$. 
The aforementioned model is an example of linear regression models with ``time-varying'' support, which have been previously considered in the literature \cite{vaswani}.

In order to analyze this model in our general framework, we remark that $\bS^{(0)}$ can be considered a latent observation model parameter that is constant across $t=1,\ldots,T$, corresponding to $\bS$ in our setup in Section \ref{sec:setup}. The ``noise'' in $\bS^{(t)}$, $V^{(t)}$, can simply be incorporated into the observation model $P(Y | X_S, \bS, S)$. A straightforward analysis of the mutual information $I_\St(b)$ using Jensen's inequality arguments similar to the proof of Theorem \ref{thm:missing_cs} (omitted here for brevity) reveals that a lower bound on $I_\St(b)$ is $\frac{1}{2} \log \left( 1 + \frac{\|\beta_\SSt\|^2 \sigma_x^2}{\sigma_w^2 + K \sigma_x^2 \sigma_v^2}\right)$. Thus, we can obtain an upper bound on the number of measurements similar to \eqref{eq:cs_ub}, where we show the effect of noise in $\bS^{(t)}$ to be equivalent to measurement noise with variance $\sigma_w^2 + K \sigma_x^2 \sigma_v^2$ as opposed to $\sigma_w^2$.

\begin{remark}
We remark that the mutual information analysis can be easily performed for different distributions (other than Gaussian) on the sensing matrix elements $X$ and the measurement noise $W$. It is only necessary to compute the mutual information $I_\St(b) = I(X_\SSt ; Y | X_\St, \bS = b, S)$ for the different probability distributions and ensure that the smoothness conditions (RC1-2) hold, if an upper bound for scaling models is desired. This is another advantage to our unifying framework, since problem-specific approaches need significantly different analyses to extend to different distributions of sensing matrices and measurement noise.
\end{remark}

\subsection{Multivariate Regression}
\label{subsec:multi}

In this problem, we consider the following linear model \cite{negahban}, where we have a total of $R$ linear regression problems,
\begin{equation*}
  \Y_{\{r\}} = \X_{\{r\}} \beta_{\{r\}} + \bm{W}_{\{r\}},~~~r=1,\ldots,R.
\end{equation*}
For each $r$, $\beta_{\{r\}} \in \mathbb{R}^N$ is a $K$-sparse vector, $\X_{\{r\}} \in \mathbb{R}^{T \times N}$ and $\Y_{\{r\}} \in \mathbb{R}^T$. The relation between different tasks is that $\beta_{\{r\}}$ have joint support $S$. This setup is also called multiple linear regression or distributed compressive sensing \cite{dcs} and is useful in applications such as multi-task learning \cite{jalali}.

\begin{figure}
  \centering
  \includegraphics[width=0.9\textwidth]{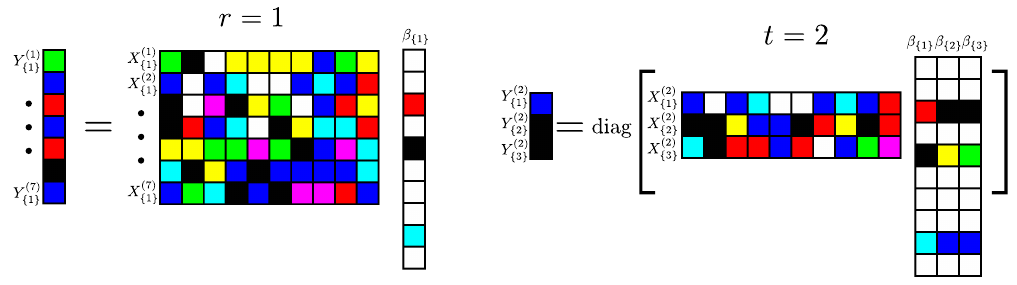}
  \caption{Mapping the multiple linear regression problem to a vector-valued outcome and variable model. On the left is the representation for a single problem $r=1$. On the right is the corresponding vector formulation, shown for sample index $t=2$.}
  \label{fig:multi}
\end{figure}

It is easy to see that this problem can be formulated in our sparse recovery framework, with vector-valued outcomes $Y$ and variables $X$. Namely, let $Y = (Y_{\{1\}}, \ldots, Y_{\{R\}}) \in \mathbb{R}^R$ be a vector-valued outcome, $X = (X_{\{1\}}^\top, \ldots, X_{\{R\}}^\top)^\top \in \mathbb{R}^{R \times N}$ be the collection of $N$ vector-valued variables and $\beta = (\beta_{\{1\}}, \ldots, \beta_{\{R\}}) \in \mathbb{R}^{N \times R}$ be the collection of $R$ sparse vectors sharing support $S$, making it block-sparse. This mapping is illustrated in Figure \ref{fig:multi}. Assuming independence between $X_{\{r\}}$ and support coefficients $\beta_{\{r\},S}$ across $r = 1,\ldots,R$, we have the following observation model:
\begin{align*}
P (Y | X,S) = p(Y | X_S) = \prod_{r=1}^R p(Y_{\{r\}} | X_{\{r\},S}) = \prod_{r=1}^R \int_{\mathbb{R}^K} p(Y_{\{r\}} | X_{\{r\},S}, \beta_{\{r\},S}) p(\beta_{\{r\},S}) \dx \beta_{\{r\},S}.
\end{align*}

We state the following theorem for the specific linear model in Section \ref{subsec:lcs}, as a direct result of Theorem \ref{thm:lcs} and the fact that the joint mutual information decomposes to $R$ identical mutual information terms in view of the equality above.

\begin{theorem}
  The lower and upper sample complexity bounds $T_\text{multi}$ per task for the linear multi-regression model above are $\frac{T_o}{R}$, where $T_o$ are the corresponding sample complexity bounds \eqref{eq:cs_lb} and \eqref{eq:cs_ub} in Theorem \ref{thm:lcs}. 
\end{theorem}

\begin{remark}
We showed that having $R$ problems with independent measurements and sparse vector coefficients decreases the number of measurements per problem by a factor of $1/R$. While having $R$ such problems increases the number of measurements $R$-fold, the inherent uncertainty in the problem is the same since the support is shared. It is then reasonable to expect such a decrease in the number of measurements.
\end{remark}

\section{Applications with Nonlinear Observations}
\label{sec:apps_nonlinear}

In this section, we consider several problems where the relationship between the input variables and the observations are nonlinear. We first look at a general framework where some of the variables are not observed, i.e., each variable is missing with some probability. We then analyze probit regression and group testing problems as other examples of problems with nonlinear observations.

\subsection{Models with Missing Features}
\label{subsec:missing}

Consider the general sparse signal processing model as described in Section \ref{sec:setup}. However, assume that instead of fully observing outcomes $\Y$ and features $\X$, we observe a $T \times N$ matrix $\Z$ instead of $\X$, with the relation
\begin{equation*}
Z_i^{(t)} = 
\left\{
\begin{array}{ll}
        X_i^{(t)}, & \mbox{w.p.\  } 1 - \rho\\
        m, & \mbox{w.p.\  } \rho
\end{array}\right. \quad \forall i, t \\
\end{equation*}
i.e., we observe a version of the feature matrix which may have entries missing with probability $\rho$, independently for each entry. We show how the sample complexity changes relative to the case where the features are fully observed. The missing data setup for specific problems have previously been considered in the literature \cite{loh,caramanis}.

First we present a universal lower bound on the number of samples for the missing data framework, by relating $I(Z_\SSt ; Y | Z_\St, \bS, S)$ to $I(X_\SSt ; Y | X_\St, \bS, S)$.

\begin{theorem} \label{thm:missing_lb}
Consider the missing data setup described above. Then we have the lower bound on the sample complexity $T_\text{miss} \geq \frac{T_o}{1-\rho}$, where $T_o$ is the lower bound on the sample complexity for the fully observed variables case given in Theorem \ref{thm:eps_lb}.
\end{theorem}

\begin{proof}
  We compute $I(Z_\SSt;Y | Z_\St,\bS = b,S)$ in terms of $I(X_\SSt;Y | X_\St,\bS = b,S)$. To do that, we compute $H(Y | Z_S, \bS=b, S)$ for any set $S$. To simplify the expressions, we omit the conditioning on $b$ and $S$ in all entropy and mutual information expressions below.
  \begin{align}
    H( Y | Z_S) = & H(Y, Z_S) - H(Z_S) \label{eq:exMD1}\\
     = & H(Y, Z_S, X_S) - H(X_S | Y, Z_S) 
     - \left( H(Z_S, X_S) - H (X_S | Z_S) \right) \label{eq:exMD2}\\
     = & H(Y | Z_S, X_S) - H(X_S | Y, Z_S) + H(X_S | Z_S) \label{eq:exMD3}\\
     = & H(Y | X_S) - H(X_S | Y, Z_S) + \sum_{k \in S} H(X_k | Z_k) \label{eq:exMD4} \\
     = & H(Y | X_S) - H(X_S | Y, Z_S) 
     + \sum_{k \in S} \big(\rho H(X_k | Z_k=m) + (1\!-\!\rho)H(X_k | Z_k\!=\!X_k)\big) \label{eq:exMD5}\\
     = & H(Y | X_S) - H(X_S | Y, Z_S) + \sum_{k \in S} \rho H(X_k) \label{eq:exMD6}\\
     = & H(Y | X_S) - H(X_S | Y, Z_S) + \rho H(X_S) \label{eq:exMD7} 
  \end{align}
  \eqref{eq:exMD1}, \eqref{eq:exMD2} and \eqref{eq:exMD3} follow from the chain rule of entropy. \eqref{eq:exMD4} follows from the conditional independence of $Y$ and $Z_S$ given $X_S$ and the independence of $Z_S, X_S$ over $k \in S$. In \eqref{eq:exMD5}, we explicitly write the conditional entropies for two values of $Z_i$. These expressions simplify to \eqref{eq:exMD6} and we group the terms over $k \in S$ to obtain \eqref{eq:exMD7}.
  
  For any set $\hat{S}$ with elements $1,\ldots,|\hat{S}|$, we can write
  \begin{align}
    H(X_{\hat{S}} | Y, Z_{\hat{S}}) & = \sum_{k=1}^{|\hat{S}|} H(X_k | Y, Z_{k,\ldots,|\hat{S}|}, X_{1,\ldots,k-1}) \label{eq:exMD8} \\
    & = \rho \sum_{k=1}^{|\hat{S}|} H(X_k | Y, Z_{k+1,\ldots,|\hat{S}|}, X_{1,\ldots,k-1}) \label{eq:exMD9}\\
    & = \rho \sum_{k=1}^{|\hat{S}|} H(X_k | Y, X_{1,\ldots,k-1}) - I(X_k ; Z_{k+1,\ldots,|\hat{S}|} | Y, X_{1,\ldots,k-1}) \label{eq:exMD10} \\
    & = \rho H(X_{\hat{S}} | Y) - \rho \sum_{k=1}^{|\hat{S}|} I(X_k ; Z_{k+1,\ldots,|\hat{S}|} | Y, X_{1,\ldots,k-1}), \label{eq:exMD11}
  \end{align}
  where \eqref{eq:exMD8} follows from the chain rule and the independence of $X_j$ and $Z_k$ given $X_k$, \eqref{eq:exMD9} by expanding the conditioning on $Z_k$, \eqref{eq:exMD10} from the definition of mutual information and \eqref{eq:exMD11} from the chain rule.
  
  W.l.o.g., assume $S = \{1, \ldots, K\}$ and $\St = \{1, \ldots, K-i\}$. Finally, using the above expressions we have
  \begin{align*}
    I(Z_\SSt;Y | Z_\St) & = H(Y | Z_\St) - H(Y | Z_S) \\
  & = H(Y | X_\St) - H(Y | X_S) + \rho (H(X_\St)  - H(X_S)) - \rho(H(X_\St | Y) - H(X_S|Y)) \\
  & \quad - \rho\left(\sum_{k=1}^K I(X_k ; Z_{k+1,\ldots,K} | Y, X_{1,\ldots,k-1}) - \sum_{k=1}^{K-i} I(X_k ; Z_{k+1,\ldots,K-i} | Y, X_{1,\ldots,k-1}) \right) \\
  & = I(X_\SSt ; Y | X_\St) + \rho (I(X_\St ; Y) - I(X_S ; Y)) \\
  & \quad - \rho\left( \sum_{k=1}^{K-i} I(X_k ; Z_{K-i+1,\ldots,K} | Y, X_{1,\ldots,k-1}, Z_{k+1,\ldots,K-i}) + \sum_{k=K-i+1}^K I(X_k ; Z_{k+1,\ldots,K} | Y, X_{1,\ldots,k-1}) \right) \\
  & \leq I(X_\SSt ; Y | X_\St) + \rho (H(Y| X_S) - H(Y | X_\St)) = (1-\rho) I(X_\SSt ; Y | X_\St).
  \end{align*}
  The first two equalities follow from the expressions we found earlier and the third equality follows from the definition of the mutual information by rearranging the sums and using the chain rule of mutual information. The last inequality follows from the non-negativity of mutual information and expanding the mutual information expressions in the first set of parentheses. The lower bound then follows from Theorem \ref{thm:eps_lb}.
\end{proof}

As a special case, we analyze the sparse linear regression model with missing data \cite{loh,caramanis}, where we obtain a model-specific upper bound on the sample complexity, in addition to the universal lower bound given by Theorem \ref{thm:missing_lb}. We consider the setup of \cite{wang,Rad-11} in the lower SNR regime $\bmin = \Theta(1/K)$, however analogous results can be shown for higher SNR regimes or for the setup of \cite{shuchin}. The proof is given in the appendix.

\begin{theorem} \label{thm:missing_cs}
  For the sparse linear regression setting of \cite{wang,Rad-11} considered in Section \ref{subsec:lcs} with $\bmin = \Theta(1/K)$ and variable matrix entries missing w.p.\ $\rho$, $T = \Omega\left(\frac{K \log N}{\log\left(1 + \frac{1-\rho}{1+\rho}\right)}\right)$ samples are sufficient for exact recovery for correlated $\bS$ and $K=O(N)$, or sufficient for recovery with a vanishing fraction of support errors for IID $\bS$ and $K = O(N/\log N)$.
\end{theorem}

\begin{remark}\label{remark:missing_cs}
We observe that the number of sufficient samples increases by a factor of $\frac{1}{\log\left(1 + \frac{1-\rho}{1+\rho}\right)}$ for missing probability $\rho$. Compare this to the upper bound given by \cite{caramanis} with scaling $\frac{1}{(1-\rho)^4}$, where the authors propose and analyze an orthogonal matching pursuit algorithm to recover the support $S$ with noisy or missing data. In this example, we have shown an upper bound that improves upon the bounds in the literature, with an intuitive universal lower bound.
\end{remark}

This example highlights the flexibility of our results in view of the mutual information characterization. This flexibility enables us to easily compute new bounds and establish new results for a very wide range of general models and their variants.

\subsection{Binary Regression}
\label{subsec:1bit}

As an example of a nonlinear observation model, we look at the following binary regression problem, also called 1-bit compressive sensing \cite{1bit,jacques,gupta} or probit regression. Regression with 1-bit measurements is interesting as the extreme case of regression models with quantized measurements, which are of practical importance in many real world applications. The conditions on the number of measurements have been studied for both noiseless \cite{jacques} and noisy \cite{gupta} models and $T=\Omega(K \log N)$ has been established as a sufficient condition for Gaussian variable matrices.

Following the problem setup of \cite{gupta}, we have
\begin{equation}
  \Y = q(\X \beta + \bm{W}),
\end{equation}
where $\X$ is a $T \times N$ matrix with IID standard Gaussian elements, and $\beta$ is an $N \times 1$ vector that is $K$-sparse with support $S$. We assume $\beta_k^2 \geq \bmin$ for $k \in S$ for a constant $\bmin$. $\bm{W}$ is a $T \times 1$ noise vector with IID standard Gaussian elements. $q(\cdot)$ is a 1-bit quantizer that outputs $1$ if the input is non-negative and $0$ otherwise, for each element in the input vector. 
This setup corresponds to the constant $\SNR$ regime in \cite{gupta}. We consider the constant $K$ regime and use the results of Theorem \ref{thm:latent_ub} to write the following theorem.

\begin{theorem}
\label{thm:qcs}
For probit regression with IID Gaussian variable matrix and the above setup, $T = \Omega(K \log N)$ measurements are sufficient to recover $S$, the support of $\beta$, with an arbitrarily small average error probability.
\end{theorem}

The proof is provided in the appendix.

\begin{remark}
Similar to linear regression, for probit regression with noise we provided a sufficiency bound that matches \cite{gupta} for an IID Gaussian matrix, for the corresponding SNR regime.
\end{remark}

We note that a lower bound on the number of measurements can also be obtained trivially, since the mutual information is upper bounded by the entropy of the measurement $Y$. Since we consider binary measurements, this leads to the lower bound $T = \Omega(K \log(N/K))$ for exact recovery through Theorem \ref{thm:eps_lb}.

\subsection{Group Testing - Boolean Model}
\label{subsec:group_testing}

\begin{figure}[tp]
  \centering
  \includegraphics[width=0.45\textwidth]{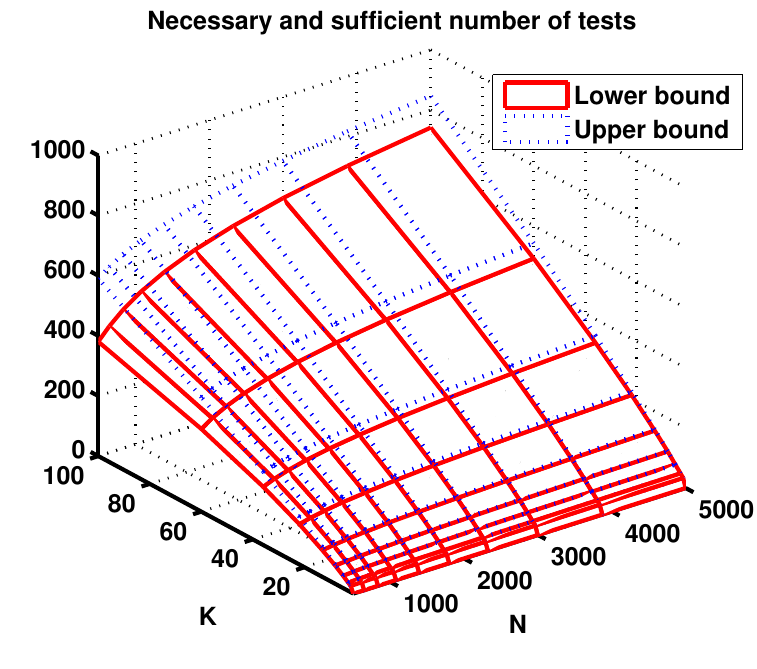}
  \caption{Upper and lower bounds on the number of tests $T$. The logarithmic dependence on $N$ and linear dependence on $K$ can be observed for large $N$. Also note that the bounds become tight as $N \to \infty$.}
  \label{fig:gt_bound}
\end{figure}

In this section, we consider another nonlinear model, namely, group testing. This problem has been covered comprehensively in \cite{group_testing} and the results derived therein can also be recovered using the generalized results we presented in this paper as we show below.

The problem of group testing can be summarized as follows. Among a population of $N$ items, $K$ unknown items are of interest. The collection of these $K$ items represents the defective set. The goal is to construct a pooling design, i.e., a collection of tests, to recover the defective set while reducing the number of required tests. In this case $\X$ is a binary measurement matrix defining the assignment of items to tests. For the noise-free case, the outcome of the tests $\Y$ is deterministic. It is the Boolean sum of the codewords corresponding to the defective set $S$, given by $\Y=\bigvee_{i\in{S}}\X_i$.

Note that for this problem there does not exist a latent observation parameter $\bS$. Therefore we have $I_{\St,0} = I(X_\SSt ; Y | X_\St, S)$ as the mutual information quantity characterizing the sample complexity for zero error recovery.

\begin{theorem}\label{thm:group_testing}
For $N$ items and $K$ defectives, the number of tests $T = \Omega(K \log (N/K))$ is necessary for $K = O(N)$ and $T = \Omega(K \log N)$ sufficient for $K=O(N)$, to identify the defective set $S$ exactly with an arbitrarily small average error probability.
\end{theorem}

We prove the theorem in the appendix. The upper and lower bounds on the number of tests (given by Theorem \ref{thm:scaling_ub} and \ref{thm:eps_lb} respectively) for the noiseless case are illustrated in Figure \ref{fig:gt_bound}. The results in \cite{group_testing} also establish upper and lower bounds on the number of tests needed for testing with additive noise (leading to false alarms) and dilution effects (leading to potential misses), as well as worst-case errors. 

We remark that we have been able to remove the extra polylog factor in $K$ in the upper bound of \cite{group_testing} and extended the regime of the upper bound from $K=o(N)$ to $K=O(N)$ with the result above.

\section{Conclusions}
\label{sec:conclusions}

We have presented a unifying framework based on noisy channel coding for analyzing sparse recovery problems. This approach unifies linear and nonlinear observation models and leads to explicit, intuitive and universal mutual information formulas for computing the sample complexity of sparse recovery problems. We explicitly focus on the inference of the combinatorial component corresponding to the support set, provably the main difficulty in sparse recovery. 
We unify sparse problems from an inference perspective based on a Markov conditional independence assumption. Our approach is not algorithmic and therefore must be used in conjunction with tractable algorithms. It is useful for identifying gaps between existing algorithms and the fundamental information limits of different sparse models. It also provides an understanding of the fundamental tradeoffs between different parameters of interest such as $K$, $N$, SNR and other model parameters. It also allows us to obtain new or improved bounds for various sparsity-based models.

\appendix

\renewcommand{\theequation}{A.\arabic{equation}}
\setcounter{equation}{0}

\subsection{First Derivative of $E_o(\rho,\bS)$ and Mutual Information}
\label{subsec:MI_derivative}

Below, we use notation for discrete variables and observations, i.e.\ sums, however the expressions are valid for continuous models if the sums are replaced by the appropriate integrals arising in the mutual information definition for continuous variables. Let $f(\rho) = \sum_{X_\Sone} Q(X_\Sone) p(Y,X_\Stwo|X_\Sone,b)^{\frac{1}{1+\rho}}$, where we omit the dependence of $f$ on $X_\Stwo$, $Y$ and $b$. Then note that $E_o(\rho,b) = -\log\left( \sum_{Y,X_\Stwo} f(\rho)^{1+\rho} \right)$. For the derivative w.r.t.\ $\rho$ we then have
\begin{equation*}
  \frac{\partial E_o(\rho,b)}{\partial \rho} = -\frac{\sum_{Y,X_\Stwo} \frac{d}{d \rho} f(\rho)^{1+\rho}}{\sum_{Y,X_\Stwo} f(\rho)^{1+\rho}} = -\frac{\sum_{Y,X_\Stwo} (1+\rho) f(\rho)^\rho \frac{d f(\rho)}{d\rho} + f(\rho)^{1+\rho} \log f(\rho)}{\sum_{Y,X_\Stwo} f(\rho)^{1+\rho}}.
\end{equation*}
Henceforth, we only consider the numerator since the denominator is obviously equal to 1 at $\rho = 0$. For the derivative of $f(\rho)$ we have
\[ \frac{d f(\rho)}{d\rho} = - \frac{\sum_{X_\Sone} Q(X_\Sone) p(Y,X_\Stwo|X_\Sone,b)^{\frac{1}{1+\rho}} \log p(Y,X_\Stwo|X_\Sone,b)}{(1+\rho)^2}, \]
therefore $\frac{d f(\rho)}{d\rho}\Big\vert_{\rho=0} = -\sum_{X_\Sone} Q(X_\Sone) p(Y,X_\Stwo|X_\Sone,b) \log p(Y,X_\Stwo|X_\Sone,b)$. For the second term, we have $f(0) \log f(0) = p(Y,X_\Stwo|b) \log p(Y,X_\Stwo|b) = \sum_{X_\Sone} Q(X_\Sone) p(Y,X_\Stwo|X_\Sone,b) \log p(Y,X_\Stwo|b)$ and using the independence of $X_\Sone$ and $X_\Stwo$ in the last equality, we can rewrite the numerator as
\begin{align}
  & \sum_{Y, X_\Stwo} \sum_{X_\Sone} Q(X_\Sone) p(Y,X_\Stwo|X_\Sone,b) \left( \log p(Y,X_\Stwo|X_\Sone,b) - \log p(Y,X_\Stwo|b) \right) \nonumber \\
  = & \sum_{Y, X_\Sone, X_\Stwo} p(Y,X_\Sone,X_\Stwo|b) \log \frac{p(Y|X_\Sone,X_\Stwo,b)}{p(Y|X_\Stwo,b)} = I(X_\Sone ; Y | X_\Stwo, \bS=b, S) = I_\Stwo(b). \label{eq:MI_b}
\end{align}

\subsection{Proof of Lemma \ref{lemma:P(Ei)} and Extension to Continuous Models}
\label{subsec:continuous}

\subsubsection*{Proof of Lemma \ref{lemma:P(Ei)}}

As we note in the main section, the proof of the lemma follows along the proof of Lemma III.1 in \cite{group_testing} which considers binary alphabets in \cite{group_testing}, yet readily generalizes to discrete alphabets for $X$ and $Y$. However, because of the latent variables $\bS$, the final step in the bottom of p.\ 1888 of \cite{group_testing} does not hold true and the proof ends with the previous equation. As a result, we have the upper bound 
\begin{equation}\label{eq:P(Ei)_ub}
  P(E_i | B) \leq e^{-\left( T \bar{E}_o(\rho) - \rho \log\binom{N-K}{i} - \log\binom{K}{i} \right)},
\end{equation}
for the multi-letter error exponent $\bar{E}_o(\rho)$ defined as
\begin{equation}\label{eq:Eo_bar}
  \bar{E}_o(\rho) = -\frac{1}{T} \log \left( \sum_\Y \sum_{\X_\St} \left[ \sum_{\X_\SSt} Q(\X_\SSt) p(\Y, \X_\St | \X_\SSt, \bS \in B)^{\frac{1}{1+\rho}} \right]^{1+\rho} \right).
\end{equation}
Note that we obtain the conditioning on $\bS \in B$ since we consider the ML decoder conditioned on that event and we write $p(\Y, \X_\St | \X_\SSt, \bS \in B) = \frac{1}{P(B)} \sum_{b \in B} p(b) p(\Y, \X_\St | \X_\SSt, b) = \frac{1}{P(B)} \sum_{b \in B} p(b) \prod_{t=1}^T p(Y^{(t)}, X_\St^{(t)} | X_\SSt^{(t)}, b)$.

Furthermore, note that Lemma \ref{lemma:P(Ei)} is missing a $\rho$ multiplying the $\frac{\log\binom{K}{i}}{T}$ term compared to Lemma III.1 in \cite{group_testing}; this is due to the fact that we do not utilize the stronger proof argument in Appendix A of \cite{group_testing}, but rather follow from the argument provided in the proof in the main body, as in \cite{gt_correction}. Using that argument, Lemma \ref{lemma:P(Ei)} can be obtained by modifying inequality (c) in p.\ 1887 that upper bounds $\Pr[E_i | \omega_0=1, \mathbf{X}_{S_1}, \Y]$ such that $\binom{N-K}{i}$ is replaced with $\binom{N-K}{i}^\rho$.

To obtain Lemma \ref{lemma:P(Ei)}, we relate this multi-letter error exponent to the single-letter error exponent $E_o(\rho)$ which we defined in \eqref{eq:Eo_lb}. The following lower bound removes the dependence between samples $t$ by considering worst-case $\bS$ and reduces it to a single-letter expression. 

\begin{lemma}\label{lemma:Eo_lb}
  \[ \bar{E}_o(\rho) \geq E_o(\rho) \triangleq \inf_{\bS \in B} E_o(\rho, \bS) - \frac{\rho}{T} H_{\frac{1}{1+\rho}}(\bS | B). \]
\end{lemma}
The lemma is proved later in the appendix. The lower bound above along with \eqref{eq:P(Ei)_ub} leads to Lemma \ref{lemma:P(Ei)}.

\subsubsection*{Continuous models}

Even though the results and proof ideas that were used in Sections \ref{sec:sufficient} and \ref{sec:necessary} are fairly general, the proof of Lemma \ref{lemma:P(Ei)} holds for discrete variables and outcomes using the proof above. In this section, we make the necessary generalizations to state an analogue of Lemma \ref{lemma:P(Ei)} for continuous variable and observation models, specifically for the case ${\cal X} = {\cal Y} = \mathbb{R}$. The extension to finite dimensional real coordinate spaces follows through the same analysis as well. We follow the methodology in \cite{gallager} and \cite{gallagerChapter}.

To simplify the exposition, we consider the extension to continuous variables in the special case of fixed and known $\bS$. In that case, $E_o(\rho)$ as defined in \eqref{eq:Eo_lb} reduces to
\begin{equation} \label{eq:Eo_simple}
  E_o(\rho) = -\log\sum_{Y,X_\St}\left(\sum_{X_\SSt}Q(X_\SSt) p(Y,X_\St|X_\SSt)^{\frac{1}{1+\rho}}\right)^{1+\rho}
\end{equation}
for $0 \leq \rho \leq 1$ with $\frac{\partial E_o(\rho)}{\partial\rho}\Big\vert_{\rho=0} = I(X_\SSt;X_\St,Y|S) = I(X_\SSt ; Y | X_\St,S)$.

We assume a continuous and bounded joint probability density function $Q(X)$ with joint cumulative distribution function $F$. The conditional probability density $p(Y=y|X_S=x)$ for the observation model is assumed to be a continuous and bounded function of both $x$ and $y$. 

Let $X' \in {\cal X}'^N$ be the random vector and $Y' \in {\cal Y}'$ be the random variable generated by the quantization of $X \in {\cal X}^N = \mathbb{R}^N$ and $Y \in {\cal Y} = \mathbb{R}$, respectively, where each variable in $X$ is quantized to $L$ values and $Y$ quantized to $J$ values. Let $F'$ be the joint cumulative distribution function of $X'$. As before, let $\hat{S}(\X,\Y)$ be the ML decoder with continuous inputs with probability of making $i$ errors in decoding denoted by $P(E_i)$. Let $\hat{S}(\X', \Y')$ be the ML decoder that quantizes inputs $\X$ and $\Y$ to $\X'$ and $\Y'$, and have the corresponding probability of error $P'(E_i)$. Define
\begin{equation*}
  E_o(\rho, X', Y') = -\log\sum_{y' \in {\cal Y}'}\sum_{x'_\St \in {\cal X}'^{K-i}}\left[\sum_{x'_\SSt \in {\cal X}'^i}Q(x'_\SSt) p(y',x'_\St|x'_\SSt)^{\frac{1}{1+\rho}}\right]^{1+\rho},
\end{equation*}
\begin{equation*}
  E_o(\rho, X, Y) = -\log \int_{\cal Y} \int_{{\cal X}^{K-i}} \left[ \int_{{\cal X}^i} Q(x_\SSt) p(y , x_\St | x_\SSt)^{\frac{1}{1+\rho}} \dx x_\SSt \right]^{1+\rho} \dx x_\St  \dx y.
\end{equation*}
where the indexing denotes the random variates that the error exponents are computed with respect to.

Utilizing Lemma \ref{lemma:P(Ei)} for discrete models, we will show that an analogue of the lemma holds for the continuous model, i.e., 
\begin{equation}
\label{eq:cont_bound}
  P(E_i)\leq e^{-\left(T E_o(\rho, X, Y) - \rho\log\binom{N-K}{i} - \log\binom{K}{i}\right)}.
\end{equation}
The rest of the proof of Theorem \ref{thm:latent_ub} will then follow as in the discrete case, by noting that $\frac{\partial E_o(\rho, X, Y)}{\partial\rho}\Big\vert_{\rho=0} = I(X_\SSt; Y | X_\St, S)$, with the mutual information definition for continuous variables \cite{coverbook}.

Our approach can be described as follows. We will increase the number of quantization levels for $Y'$ and $X'$, respectively. Then, since the discrete result in \eqref{eq:P(Ei)} holds for any number of quantization levels, by taking limits we will be able to show that
\begin{equation}
  P'(E_i) \leq e^{-\left(T E_o(\rho, X, Y) - \rho\log\binom{N-K}{i} - \log\binom{K}{i}\right)}.
\end{equation}

Since $\hat{S}(\X,\Y)$ is the minimum probability of error decoder, any upper bound for $P'(E_i)$ will also be an upper bound for $P(E_i)$, thereby proving \eqref{eq:cont_bound}.

Assume $Y$ is quantized with the quantization boundaries denoted by $a_1, \ldots, a_{J-1}$, with $Y' = a_j$ if $a_{j-1} < Y \leq a_j$. For convenience denote $a_0 = -\infty$ and $a_J = \infty$. Furthermore, assume the quantization boundaries are equally spaced, i.e. $a_j - a_{j-1} = \Delta_J$ for $2 \leq j \leq J-1$. Now, we have that
\begin{alignat}{2}
  E_o(\rho, X', Y') = & -\log & & \sum_{j=1}^J \sum_{x'_\St} \left[ \sum_{x'_\SSt}Q(x'_\SSt) \left( \int_{a_{j-1}}^{a_j} p(y,x'_\St|x'_\SSt) \dx y \right)^{\frac{1}{1+\rho}}\right]^{1+\rho} \\ 
  = & -\log & & \Bigg\{ \sum_{j=2}^{J-1} \Delta_J \sum_{x'_\St} \left[ \sum_{x'_\SSt}Q(x'_\SSt) \left( \frac{\int_{a_{j-1}}^{a_j} p(y,x'_\St|x'_\SSt) \dx y}{\Delta_J} \right)^{\frac{1}{1+\rho}}\right]^{1+\rho} \\
  & & & \; + \sum_{x'_\St} \left[ \sum_{x'_\SSt}Q(x'_\SSt) \left( \int_{-\infty}^{a_1} p(y,x'_\St|x'_\SSt) \dx y \right)^{\frac{1}{1+\rho}}\right]^{1+\rho} \\
  & & & \; + \sum_{x'_\St} \left[ \sum_{x'_\SSt}Q(x'_\SSt) \left( \int_{a_{J-1}}^{\infty} p(y,x'_\St|x'_\SSt) \dx y \right)^{\frac{1}{1+\rho}}\right]^{1+\rho} \Bigg\}.
\end{alignat}

Let $J \to \infty$ and for each $J$ choose the sequence of quantization boundaries such that $\lim \Delta_J = 0$, $\lim a_{J-1} = \infty$, $\lim a_1 = -\infty$. Then, the last two terms disappear and using the fundamental theorem of calculus, we obtain
\begin{equation}
  \lim_{J \to \infty} E_o(\rho, X', Y') = E_o(\rho, X', Y) = -\log \int_{\cal Y} \sum_{x'_\St} \left[ \sum_{x'_\SSt}Q(x'_\SSt) p(y,x'_\St|x'_\SSt)^{\frac{1}{1+\rho}}\right]^{1+\rho} \dx y.
\end{equation}

It can also be shown that $E_o(\rho, X', Y')$ increases for finer quantizations of 
$Y'$, therefore $E_o(\rho, X', Y)$ gives the smallest upper bound over $P'(E_i)$ over the quantizations of $Y$, similar to \cite{gallager}. However, this is not necessary for the proof.

We repeat the same procedure for $X$. Assume each variable $X_n$ in $X$ is quantized with the quantization boundaries denoted by $b_1, \ldots, b_{L-1}$, with $X_n' = b_l$ if $b_{l-1} < X_n \leq b_l$. For convenience denote $b_0 = -\infty$ and $b_L = \infty$. Furthermore, assume that the quantization boundaries are equally spaced, i.e. $b_l - b_{l-1} = \Delta_L$ for $2 \leq l \leq L-1$. Then, we can write
\begin{alignat}{2}
  E_o(\rho, X', Y) = & -\log \int_Y && \sum_{l=1}^L \left[ \sum_{x'_\SSt}Q(x'_\SSt) \left( \int_{b_{l-1}}^{b_l} p(y,x_\St|x'_\SSt) \dx x_\St \right)^{\frac{1}{1+\rho}} \right]^{1+\rho} \dx y \\
  = & -\log \int_{\cal Y} && \sum_{l=1}^L \left[ \int_{{\cal X}^i} \left( \int_{b_{l-1}}^{b_l} p(y,x_\St|x_\SSt) \dx x_\St \right)^{\frac{1}{1+\rho}} \dx F'(x_\SSt) \right]^{1+\rho} \dx y \label{eq:cont_xs1} \\
  = & -\log \int_{\cal Y} && \Bigg\{ \sum_{l=2}^{L-1} \Delta_L \left[ \int_{{\cal X}^i} \left( \frac{\int_{b_{l-1}}^{b_l} p(y,x_\St|x_\SSt) \dx x_\St}{\Delta_L} \right)^{\frac{1}{1+\rho}} \dx F'(x_\SSt) \right]^{1+\rho} \nonumber \\
  &&& \; + \int_{{\cal X}^i} \left( \int_{-\infty}^{b_1} p(y,x_\St|x_\SSt) \dx x_\St \right)^{\frac{1}{1+\rho}} \dx F'(x_\SSt) \nonumber \\
  &&& \; + \int_{{\cal X}^i} \left( \int_{b_{L-1}}^{\infty} p(y,x_\St|x_\SSt) \dx x_\St \right)^{\frac{1}{1+\rho}} \dx F'(x_\SSt)  \Bigg\} \dx y, \label{eq:cont_xs2}
\end{alignat}
where \eqref{eq:cont_xs1} follows with $F'(x_\SSt)$ being the step function that represents the cumulative distribution function of the quantized variables $X'_\SSt$.

Let $L \to \infty$, for each $L$ choose a set of quantization points such that $\lim \Delta_L = 0$, $\lim b_{L-1} = \infty$, $\lim b_1 = -\infty$. Again, the second and third terms disappear and the first sum converges to the integral over $X_\St$. Note that $p(y,x_\St|x_\SSt)$ is a bounded continuous function of all its variables since it was assumed that $Q(x)$ and $p(y|x)$ were bounded and continuous. Also note that $\lim_{L \to \infty} F' = F$, which implies the weak convergence of the probability measure of $X'$ to the probability measure of $X$. Given these facts, using the portmanteau theorem we obtain that $E_{F'} \left[ p(Y,X_\St|X_\SSt)\right] \to E_F\left[ p(Y,X_\St|X_\SSt) \right]$, which leads to
\begin{equation}
  \lim_{L \to \infty} E_o(\rho, X', Y) = -\log \int_{\cal Y} \int_{{\cal X}^{K-i}} \left[ \int_{{\cal X}^i} p(y , x_\St | x_\SSt)^{\frac{1}{1+\rho}} \dx F(x_\SSt) \right]^{1+\rho} \dx x_\St \dx y = E_o(\rho, X, Y).
\end{equation}

This leads to the following result, completing the proof.
\begin{equation}
  P(E_i) \leq P'(E_i) \leq \lim_{J, L \to \infty} e^{-\left(T E_o(\rho, X', Y') - \rho\log\binom{N-K}{i} - \log\binom{K}{i} \right)} = e^{-\left(T E_o(\rho, X, Y) - \rho\log\binom{N-K}{i} - \log\binom{K}{i} \right)}.
\end{equation}

\subsection{Proof of Lemma \ref{lemma:Eo_lb}}
\label{subsec:proof_Eo_lb}

As in the previous proofs, while we use notation for discrete variables and observations, the proof below is valid for continuous models. We omit the conditioning on $\bS \in B$ and let $p(\bS)$ denote the probability distribution and $H_\alpha(\bS)$ the R\'enyi entropy of $\bS$ conditioned on $\bS \in B$ w.l.o.g. 

For the error exponent $\bar{E}_o(\rho)$ as defined in \eqref{eq:Eo_bar}, let 
\[ g_\rho(\Y,\X_\St) = E_{\X_\SSt} \left[ E_{\bS} \left[ p(\Y, \X_\St | \X_\SSt, \bS) \right]^{\frac{1}{1+\rho}} \right]^{1+\rho}, \]
such that $\bar{E}_o(\rho) = -\frac{1}{T}\log \sum_{\Y, \X_\St} g_\rho(\Y, \X_\St)$. We then write the following chain of inequalities:
\begin{align*}
  \sum_{\Y, \X_\St} g_\rho(\Y, \X_\St) & = \sum_{\Y, \X_\St} E_{\X_\SSt} \left[ \left( \sum_{\bS} p(\bS) p(\Y, \X_\St | \X_\SSt, \bS) \right)^{\frac{1}{1+\rho}} \right]^{1+\rho} \\
  & \leq \sum_{\Y, \X_\St} E_{\X_\SSt} \left[ \sum_{\bS} p(\bS)^{\frac{1}{1+\rho}} p(\Y, \X_\St | \X_\SSt, \bS)^{\frac{1}{1+\rho}} \right]^{1+\rho} \\
  & = \sum_{\Y, \X_\St}  \left( \sum_{\bS} p(\bS)^{\frac{1}{1+\rho}} E_{\X_\SSt} \left[p(\Y, \X_\St | \X_\SSt, \bS)^{\frac{1}{1+\rho}}\right] \right)^{1+\rho} \\
  & \leq R_\rho^{1+\rho} \sum_{\Y, \X_\St} \left( \sum_{\bS} p'(\bS) E_{\X_\SSt} \left[p(\Y, \X_\St | \X_\SSt, \bS)^{\frac{1}{1+\rho}}\right] \right)^{1+\rho} \\
  & \leq R_\rho^{1+\rho} \sum_{\Y, \X_\St} \sum_{\bS} p'(\bS) E_{\X_\SSt} \left[p(\Y, \X_\St | \X_\SSt, \bS)^{\frac{1}{1+\rho}}\right]^{1+\rho} \\
  & = R_\rho^{1+\rho} \sum_{\bS} p'(\bS) \sum_{\Y, \X_\St} E_{\X_\SSt} \left[p(\Y, \X_\St | \X_\SSt, \bS)^{\frac{1}{1+\rho}}\right]^{1+\rho} \\
  & = R_\rho^{1+\rho} \sum_{\bS} p'(\bS) \left[ \sum_{Y, X_\St} \left( E_{X_\SSt} \left[p(Y, X_\St | X_\SSt, \bS)^{\frac{1}{1+\rho}}\right] \right)^{1+\rho} \right]^T,
\end{align*}
where the first inequality follows from the subadditivity of exponentiating with $\frac{1}{1+\rho}$ and the second follows by multiplying and dividing inside the sum by $R_\rho = \sum_{\bS} p(\bS)^\frac{1}{1+\rho}$ and defining $p'(\bS) = \frac{p(\bS)^{\frac{1}{1+\rho}}}{R_\rho}$. The third inequality follows using Jensen's inequality. We obtain the final expression by noting that the expression in the square brackets factorizes over $t = 1,\ldots,T$ and is IID over $t$ when conditioned on $\bS$.

Noting that $\log R_\rho = \frac{\rho}{1+\rho} H_{\frac{1}{1+\rho}}(\bS)$, where $H_\alpha(\cdot)$ is the R\'enyi entropy of order $\alpha$, we then have 
\begin{align*}
  \bar{E}_o(\rho) & \geq -\frac{1+\rho}{T}\log R_\rho - \frac{1}{T} \log \sum_{\bS} p'(\bS) \left[ \sum_{Y, X_\St} E_{X_\SSt} \left[p(Y, X_\St | X_\SSt, \bS)^{\frac{1}{1+\rho}}\right]^{1+\rho} \right]^T \\
  & = -\frac{\rho}{T} H_{\frac{1}{1+\rho}}(\bS) - \frac{1}{T} \log \sum_{\bS} p'(\bS) \left[ \sum_{Y, X_\St} E_{X_\SSt} \left[p(Y, X_\St | X_\SSt, \bS)^{\frac{1}{1+\rho}}\right]^{1+\rho} \right]^T \\
  & \geq -\frac{\rho}{T} H_{\frac{1}{1+\rho}}(\bS) - \frac{1}{T} \log \sup_{\bS} \left[ \sum_{Y, X_\St} E_{X_\SSt} \left[p(Y, X_\St | X_\SSt, \bS)^{\frac{1}{1+\rho}}\right]^{1+\rho} \right]^T \\
  & = -\frac{\rho}{T} H_{\frac{1}{1+\rho}}(\bS) + \inf_{\bS} \left(-\log \left[ \sum_{Y, X_\St} E_{X_\SSt} \left[p(Y, X_\St | X_\SSt, \bS)^{\frac{1}{1+\rho}}\right]^{1+\rho} \right] \right) \\
  & = -\frac{\rho}{T} H_{\frac{1}{1+\rho}}(\bS) + \inf_{\bS} E_o(\rho, \bS) = E_o(\rho).
\end{align*}

\subsection{Proof of Lemma \ref{lemma:F_N_bd}}

Let ${\cal D} \subset {\cal C}$ be a countable dense subset and w.l.o.g.\ write ${\cal D} = \{s_i: i \in \mathbb{N}\}$. 
We also note that since $F_N(\rho,s)$ is monotone in $\rho$ and its limit $F_\infty(\rho,s) \triangleq \lim_{N \to \infty} F_N(\rho,s)$ is continuous w.r.t.\ $\rho$, it follows that $F_N(\rho,s)$ is equicontinuous in $\rho$.

Let $\epsilon_1 > 0$. Then, for all $s_i$, there exists $\rho_i > 0$ such that $F_N(\rho_i,s_i) \geq F_N(0,s_i) - \epsilon_1 = 1 - \epsilon_1$ for all $N$ uniformly due to the aforementioned equicontinuity. $\rho_i$ depends on $s_i$ but not on $N$ due to equicontinuity.

Let $\epsilon_2 > 0$. Then, for all $s_i$, there exists $\delta_i > 0$ such that $|F_\infty(\rho_i,s) - F_\infty(\rho_i,s_i)| < \epsilon_2$ for all $s$ such that $|s - s_i| < \delta_i$ due to the continuity of $F_\infty(\rho,s)$ w.r.t.\ $s$. $\delta_i$ depends on $s_i$ and $\rho_i$, which in turn depends only on $s_i$.

Let $\epsilon_3 > 0$. Then, for all $s_i$, there exists $N_i$ such that for all $N \geq N_i$, $|F_N(\rho_i,s) - F_\infty(\rho_i,s)| < \frac{\epsilon_3}{2}$ for all $s$ uniformly due to the uniform convergence of $F_N(\rho_i,s)$ to $F_\infty(\rho_i,s)$ w.r.t.\ $s$. $N_i$ depends on $\rho_i$, which in turn depends only on $s_i$. This implies that for $N \geq N_i$,
\begin{align*}
  |F_N(\rho_i,s) - F_N(\rho_i,s_i)| & = |F_N(\rho_i,s) - F_\infty(\rho_i,s) + F_\infty(\rho_i,s) - F_\infty(\rho_i,s_i) + F_\infty(\rho_i,s_i) - F_N(\rho_i,s_i)| \\
  & \leq |F_N(\rho_i,s) - F_\infty(\rho_i,s)| + |F_\infty(\rho_i,s) - F_\infty(\rho_i,s_i)| + |F_\infty(\rho_i,s_i) - F_N(\rho_i,s_i)| \\
  & < \frac{\epsilon_3}{2} + \epsilon_2 + \frac{\epsilon_3}{2} = \epsilon_2 + \epsilon_3,
\end{align*}
which follows from the triangle inequality.

Define the collection of sets $C = \left\{ B_{\delta_i}(s_i): i \in \mathbb{N} \right\}$, where $B_{\delta_i}(s_i)$ is the open ball with radius $\delta_i$ and center $s_i \in {\cal D}$. Note that $C$ is an open cover of ${\cal C}$ since $\delta_i > 0$ for all $i$ and we chose ${\cal D}$ to be a dense subset of ${\cal C}$. Since ${\cal C}$ is compact, it then follows that there exists a finite subcover $\bar{C}$ of $C$. W.l.o.g.\ let $\bar{C} = \{s_i: i \in \{1,\ldots,p\}\}$ for a constant integer $p$. Also, define the constants $\rho_c = \min_{i = 1,\ldots,p} \rho_i$ and $N_0 = \max_{i = 1,\ldots,p} N_i$, for which we have $\rho_c > 0$ and $N_0 < \infty$. 

Let $s \in {\cal C}$. Since $\bar{C}$ is a finite cover of ${\cal C}$, there exists $i \in \{1,\ldots,p\}$ such that $s \in B_{\delta_i}(s_i)$. This implies that for all $N \geq N_0 \geq N_i$,
\[ F_N(\rho_c,s) \geq F_N(\rho_i,s) \geq F_N(\rho_i,s_i) - (\epsilon_2 + \epsilon_3) \geq 1 - (\epsilon_1 + \epsilon_2 + \epsilon_3), \]
where in the first inequality we used the fact that $F_N(\rho,s)$ is non-decreasing in $\rho$ and $\rho_c \leq \rho_i$. Choosing $\epsilon_1,\epsilon_2,\epsilon_3$ such that $\epsilon_1+\epsilon_2+\epsilon_3 = c$ proves the lemma.

The necessity of compactness and uniform convergence assumptions are apparent from the proof. Without compactness, we cannot find a finite subcover and if we selected $\rho_c = \inf_{i \in \mathbb{N}} s_i$ and $N_0 = \sup_{i \in \mathbb{N}} N_i$, we are not guaranteed that $\rho_c > 0$ and $N_0 < \infty$. Similarly, without uniform convergence the sequence index $N_i(s)$ for which $|F_N(\rho_i,s) - F_\infty(\rho_i,s)| < \frac{\epsilon_3}{2}$ for $N \geq N_i(s)$ would depend on $s$ and we cannot find a $N_i$ to upper bound $N_i(s)$ over all $s \in B_{\delta_i}(s_i)$.

\subsection{Proof of Theorem \ref{thm:partial}}
\label{subsec:partial_proof}

We first distingiush between two scaling regimes, namely $K = O(\log N)$ and $K = \omega(\log N)$. For the first regime, we remark that $H_\frac{1}{2}(\bS) = O(K)$ is always asymptotically dominated by $\log\binom{N-K}{K-|\St|}$, thus the condition \eqref{eq:scaling_ub} in Theorem \ref{thm:scaling_ub} reduces to \eqref{eq:partial_ub} in Theorem \ref{thm:partial}. This implies that we are able to obtain exact recovery, i.e.\ $\lim_{N \to \infty} P(E) = 0$ in the setup considered in Theorem \ref{thm:partial} if $K = O(\log N)$. Thus in the rest of the proof we deal with the case $K = \omega(\log N)$.

Consider a sequence of numbers $\alpha_N \in (0,1)$, representing the fraction of errors in the support that we would like to allow for each $N$. Define the corresponding sequence of error events $E_{\alpha_N}$, where each $E_{\alpha_N}$ is the event that the recovered set has more than $\alpha_N K$ errors, i.e., $|S \setminus \hat{S}_N(\X, \Y)| > (1-\alpha_N) K$. We assume for notational convenience that $\alpha_N K$ correspond to integers.

From the end of proof of Theorem \ref{thm:scaling_ub}, we have that for $T$ satisfying condition \eqref{eq:scaling_ub},
\[ K P(E_i) \leq \exp\left(-c' D_N\right) \leq c \frac{1}{\binom{N-K}{i}}, \]
for constants $c, c'$. Since we have $P(E_{\alpha_N}) \leq \sum_{i=\alpha_N K + 1}^K P(E_i)$, it then follows that
\[ P(E_{\alpha_N}) \leq (1-\alpha_N) K \max_{i=\alpha_N K + 1,\ldots,K} P(E_i) \leq c \frac{1}{\binom{N-K}{\alpha_N K}}. \]

Let $\alpha_N = \frac{1}{\log \log N}$. For this $\alpha_N$ and $i \geq \alpha_N K + 1$, we have
\[ \log\binom{N-K}{K-|\St|} = \log\binom{N-K}{i} = \log\binom{N-K}{K/\log \log N} = \Theta\left( \frac{K}{\log\log N} \log \left(\frac{N \log\log N}{K}\right)\right) = \Omega(K), \]
since $K = O(N/\log N)$. Thus for this choice of $\alpha_N$ we again have that $H_\frac{1}{2} = O(K)$ is asymptotically dominated by $\log\binom{N-K}{K-|\St|}$. This implies that the above bound on $P(E_{\alpha_N})$ can be achieved with the condition \eqref{eq:partial_ub} on $T$ rather than \eqref{eq:scaling_ub}.

Investigating the upper bound on the error probability for this choice of $\alpha_N = \frac{1}{\log\log N}$, we have
\[ P(E_{\alpha_N}) \leq c \frac{1}{\binom{N-K}{\alpha_N K}} = \Theta\left( \left(\frac{K}{N \log\log N}\right)^\frac{K}{\log\log N} \right) = O( N^{-(1+q)} ). \]
for some constant $q > 0$. To see that the last equivalence holds, take the $-\log$ of both sides, where for the left-hand side we have $\Theta\left( \frac{K}{\log\log N} \log \left( \frac{N \log\log N}{K} \right)\right)$. If we show that this term is lower bounded by $c' + (1+q)\log N$ for any constant $c'$ it implies that the equivalence holds. Using the fact that $K = \omega(\log N)$, we can first lower bound this term by $\frac{C \log N}{\log\log N} \log \left( \frac{N \log\log N}{K} \right)$ for any (arbitrarily large) constant $C$. Then, since $K = O(N/ \log N)$ we can again lower bound the term by $\frac{C \log N}{\log\log N} \log \left( \log N \log\log N \right) = \frac{C \log N}{\log\log N} (\log\log N + \log\log\log N) \geq (1+q) \log N$, proving the equivalence.

With this scaling on $P(E_{\alpha_N})$, we then have that $\sum_{N=1}^{\infty} P(E_{\alpha_N}) \leq \infty$ and it follows from the Borel-Cantelli lemma that $\Pr[\limsup_{N \to \infty} E_{\alpha_N}] = 0$. Writing the $\limsup$ explicitly, we have
\[ \limsup_{N \to \infty} E_{\alpha_N} = \bigcup_{N=1}^\infty \bigcap_{M \geq N}^\infty E_{\alpha_M} = \left\{\forall N, \exists M \geq N~\mathrm{s.t.}~\frac{|S \setminus \hat{S}_M(\X,\Y)|}{K} \geq \alpha_M \right\},\]
thus with probability one $\limsup_{N \to \infty} \frac{|S \setminus \hat{S}_N(\X,\Y)|}{K} = 0$ and the theorem follows.

\subsection{Sparse Linear Regression Analysis}
\label{subsec:lcs_analysis}

\subsubsection*{Error exponent $E_N(\rho,b)$ and regularity conditions}

We now compute the error exponent $E_N(\rho,b)$ similar to \cite{aistats}, where $b = (\bone, \btwo)$ and $(\bone, \btwo) = (b_\Sone, b_\Stwo)$ where $|\Sone| = i$. While we can obtain an upper bound directly using Lemma \ref{lemma:P(Ei)} and the computed error exponent, we use this computation to show that the regularity conditions (RC1-2) hold for Theorem \ref{thm:scaling_ub}.

We can write 
\[ E_N(\rho,b) = -\log \left( \int_Y E_{X_\Stwo} \left[ E_{X_\Sone} \left[ p(Y | X_S, b)^{\frac{1}{1+\rho}} \right]^{1+\rho} \right] \dx Y \right) \]
and below we compute the inner expectation over $X_\Sone$.

\begin{align*}
\int_{X_\Sone} & P(X_\Sone) p(Y|X_\Sone,X_\Stwo, \bS)^{\frac{1}{1+\rho}} \dx X_\Sone = \int_{\mathbb{R}^i} {\cal N}\left( x; 0, \sigma_x^2 I_i \right) {\cal N}\left(y - x^\top \bone - x_2^\top \btwo; 0, \sigma_w^2 \right)^{\frac{1}{1+\rho}} \dx x \\
& = \left( \frac{1}{\sqrt{2\pi A}}\right)^i \left( \frac{1}{\sigma_w \sqrt{2\pi}} \right)^{\frac{1}{1+\rho}} \int_{\mathbb{R}^i} \exp \left( - \frac{x^\top x}{2 A} \right) \exp \left( - \frac{(y - x^\top \bone - x_2^\top \btwo)^2}{2B} \right) \dx x \\
& = \left( \frac{1}{\sqrt{2\pi A}}\right)^i \left( \frac{1}{\sigma_w \sqrt{2\pi}} \right)^{\frac{1}{1+\rho}} \int_{\mathbb{R}^i} \exp \left( - \frac{x^\top x}{2 A}  - \frac{(x^\top \bone + C)^2}{2B} \right) \dx x\\
& = \left( \frac{1}{\sqrt{2\pi A}}\right)^i \left( \frac{1}{\sigma_w \sqrt{2\pi}} \right)^{\frac{1}{1+\rho}} \int_{\mathbb{R}^i} \exp \left( - \frac{1}{2}(x + (BD)^{-1} AC \bone)^\top \frac{D}{A} (x + (BD)^{-1} AC \bone) \right) \exp \left( - \frac{C^2}{2E} \right) \dx x
\end{align*}
where $A = \sigma_x^2$, $B = \sigma_w^2 (1+\rho)$, $C = x_2^\top\btwo - y$, $D = I_i + \frac{A}{B} \bone \bone^\top$ and $E = \frac{B}{1 - \frac{A}{B} \bone^\top D^{-1} \bone}$. Then taking the integral, some terms on the left cancel and we have
\begin{equation} \label{eq:cs_pyx}
  \int_{X_\Sone} P(X_\Sone) p(Y|X_\Sone,X_\Stwo, \bS)^{\frac{1}{1+\rho}} \dx X_\Sone = \left( \frac{1}{\sigma_w \sqrt{2\pi}} \right)^{\frac{1}{1+\rho}} \frac{1}{\sqrt{|D|}} \exp \left( - \frac{C^2}{2E} \right).
\end{equation}

Writing the second integral that is over $X_\Stwo$, we then have
\begin{align*}
\int_{X_\Stwo} & P(X_\Stwo) \left[ \int_{X_\Sone} P(X_\Sone) p(Y|X_\Sone,X_\Stwo, \bS)^{\frac{1}{1+\rho}} \dx X_\Sone \right]^{1+\rho} \dx X_\Stwo \\  
& = \sqrt{\frac{1}{\sigma_w^2 2\pi}} \frac{1}{\sqrt{|D|}^{(1+\rho)}} \int_{\mathbb{R}^{K-i}} {\cal N}(x ; 0, A I_{K-i}) \exp \left( -\frac{(x^\top \btwo - y)^2}{2E'} \right) \dx x \\
& = \sqrt{\frac{1}{\sigma_w^2 2\pi}} \frac{1}{\sqrt{|D|}^{(1+\rho)}} \left(\frac{1}{\sqrt{2\pi A}}\right)^{K-i} \int_{\mathbb{R}^{K-i}} \exp\left( -\frac{x^\top x}{2A} - \frac{(x^\top \btwo - y)^2}{2E'}\right) \dx x \\
& = \sqrt{\frac{1}{\sigma_w^2 2\pi}} \frac{1}{\sqrt{|D|}^{(1+\rho)}} \left(\frac{1}{\sqrt{2\pi A}}\right)^{K-i} \int_{\mathbb{R}^{K-i}} \exp\left(-\frac{1}{2}(x - y (E' G)^{-1}A \btwo)^\top \frac{G}{A} (x - y (E' G)^{-1}A \btwo)\right) \exp \left( -\frac{y^2}{2H} \right) \dx x
\end{align*}
where $E' = \frac{E}{1+\rho}$, $G = 1 + \frac{A}{E'} \btwo \btwo^\top$ and $H = \frac{E'}{1 - \frac{A}{E'} \btwo^\top G^{-1} \btwo}$. Again, evaluating the integral, we obtain
\[ \int_{X_\Stwo} P(X_\Stwo) \left[ \int_{X_\Sone} P(X_\Sone) p(Y|X_\Sone,X_\Stwo, \bS)^{\frac{1}{1+\rho}} \dx X_\Sone \right]^{1+\rho} \dx X_\Stwo = \frac{1}{\sigma_w \sqrt{2\pi}} \frac{1}{\sqrt{|D|}^{(1+\rho)}} \frac{1}{\sqrt{|G|}} \exp\left(-\frac{y^2}{2H}\right). \]

Integrating the above expression w.r.t.\ $Y = y$, we see that
\[ \int_{Y} \int_{X_\Stwo} P(X_\Stwo) \left[\int_{X_\Sone} P(X_\Sone) p(Y|X_\Sone,X_\Stwo, \bS)^{\frac{1}{1+\rho}} \dx X_\Sone \right]^{1+\rho} \dx X_\Stwo \dx Y = \frac{1}{\sigma_w \sqrt{|D|}^{(1+\rho)}} \sqrt{\frac{H}{|G|}}. \] 

By the matrix determinant lemma, we have $|D| = 1 + \frac{A}{B} \bone^\top \bone$ and by the Sherman-Morrison formula, $D^{-1} = I_i - \frac{\bone \bone^\top}{\frac{B}{A} + \bone^\top \bone}$. Similarly, $|G| = 1 + \frac{A}{E'} \btwo^\top \btwo$ and $G^{-1} = I_i - \frac{\btwo \btwo^\top}{\frac{E'}{A} + \btwo^\top \btwo}$. By plugging in these expressions, we can then see that $E' = \frac{B |D|}{1+\rho}$ and $H = E' |G|$. We simplify the above expression to obtain
\begin{equation} \label{eq:cs_beta}
  \frac{1}{\sigma_w \sqrt{|D|}^{(1+\rho)}} \sqrt{\frac{H}{|G|}} = \frac{1}{\sigma_w \sqrt{|D|}^{(1+\rho)}} \sqrt{\frac{B |D|}{1+\rho}} = \left( \frac{1}{\sqrt{|D|}} \right)^\rho = \left( 1 + \frac{\sigma_x^2 \bone^\top \bone}{(1+\rho) \sigma_w^2}\right)^{-\rho/2}.
\end{equation}
and therefore letting $C_N = \frac{i \sigma_x^2}{\sigma_w^2}$ and $s = \frac{\|\bone\|^2}{i}$, we have
\begin{equation}\label{eq:Eo_lcs}
  E_N(\rho,b) = E_N(\rho,s) = \frac{\rho}{2} \log \left( 1 + \frac{C_N s}{1+\rho}\right).
\end{equation}

From above, it is now obvious that (RC1) is satisfied since $s \in {\cal C} = [\bmin, \bmax]$, which is a finite interval independent of $N$. It also follows from straightforward algebra that
\[ F_N(\rho,s) = \frac{\log\left(1 + \frac{C_N s}{1+\rho}\right) - \frac{\rho}{(1+\rho)^2}\frac{C_N s}{1+\frac{C_N s}{1+\rho}}}{\log(1 + C_N s)}. \]

We assume $C_N$ is monotonic with $N$ and consider three cases for the limit $F_\infty(\rho,s)$: $C_N \to 0$, $C_N \to \infty$ or $C_N = c$. In the first case we have $F_\infty(\rho,s) = \frac{1}{(1+\rho)^2}$ and in the second case $F_\infty(\rho,s) = 1$. In all three cases it is obvious that $F_\infty(\rho,s)$ is continuous in $\rho$ for all $s$. To prove uniform convergence in $s$, we claim that $F_N(\rho,s)$ is a monotonically non-decreasing or non-increasing sequence for each $\rho$ and $s$. For the first case, for large enough $N$ we observe that $F_N(\rho,s) \approx \frac{\frac{C_N s}{1+\rho} - \frac{\rho C_N s}{(1+\rho)(1+\rho+C_N s)}}{C_N s} = \frac{1+C_N s}{(1+\rho)(1+\rho+C_N s)}$, which is monotone in $N$. For the second case, for large enough $N$ $F_N(\rho,s) \approx \frac{\log(C_N s) - \log(1+\rho) - \frac{\rho}{1+\rho}}{\log(C_N s)} = 1 - \frac{1}{(1+\rho)C_N s}$, which is also monotone in $N$. For the third case $F_N(\rho,s)$ is constant w.r.t.\ $N$ and thus also monotone. Then by Dini's Theorem we have that $F_N(\rho,s)$ converges uniformly in $s$, proving that (RC2) holds.

\subsection{Proof of Theorem \ref{thm:missing_cs}}

Let $(\Sone, \Stwo) \triangleq (\SSt, \St)$ and for $|\Sone| = i$, define $Z_1 = Z_\Sone$ and $Z_2 = Z_\Stwo$. For simplicity of exposition, we will assume the worst-case support with any $\bS = b$ such that $b_k^2 = b_{\min}$ for $k \in S$, however the results can be generalized to random $\bS$ similar to the proof of Theorem \ref{thm:lcs}. We also omit the explicit conditioning on $S$ and $b$ in the expressions below.

To prove the theorem, we will obtain a lower bound on $I(Z_1 ; Y | Z_2) = h(Y | Z_2) - h(Y | Z_1, Z_2)$. Let $M_1 = \{k \in \Sone: X_k = m\}$ and $M_2 = \{k \in \Stwo: X_k = m\}$ denote the set of missing features in each set, then it simply follows that $Z_1 = (\{X_k\}_{k \in \Sone \cap M_1^c}, M_1)$ and $Z_2 = (\{X_k\}_{k \in \Stwo \cap M_2^c}, M_2)$.

We will start by proving an upper bound on $h(Y | Z_1, Z_2)$. We have,
\begin{align*}
  h(Y | Z_1, Z_2) & = E_{Z_1,Z_2}\left[ h( X_S^\top b + W | Z_1, Z_2) \right] \\
  & = E_{X_{\Sone,M_1^c}, X_{\Stwo,M_2^c}, M_1, M_2} \left[ h(X_S^\top b + W | X_{\Sone,M_1^c}, X_{\Stwo,M_2^c}, M_1, M_2) \right] \\
  & = E_{M_1, M_2} \left[ h(X_{M_1}^\top b_{M_1} + X_{M_2}^\top b_{M_2} + W | M_1, M_2) \right] \\
  & = E_{M_1, M_2} \left[ \frac{1}{2}\log\left(2 \pi e b_{\min} \left[ \sigma_x^2|M_1| + \sigma_x^2|M_2| + \frac{\sigma_w^2}{b_{\min}} \right] \right) \right] \\
  & \leq E_{M_2} \left[ \frac{1}{2}\log\left(2 \pi e b_{\min} \left[ \sigma_x^2 E_{M_1}[|M_1|] + \sigma_x^2 |M_2| + \frac{\sigma_w^2}{b_{\min}} \right] \right) \right] \\
  & = E_{M_2} \left[ \frac{1}{2}\log\left(2 \pi e b_{\min} \left[ \sigma_x^2 i \rho + \sigma_x^2 |M_2| + \frac{\sigma_w^2}{b_{\min}} \right] \right) \right].
\end{align*}
The first two equalities follow by expanding $Y$ and $Z_1$, $Z_2$. The third equality follows by subtracting the known quantities related to $X_{\Sone,M_1^c}, X_{\Stwo,M_2^c}$ from the entropy expression. The fourth equality follows by noting that the variable inside the entropy conditioned on $M_1$ and $M_2$ is Gaussian and then computing its variance. We use Jensen's inequality over $M_1$ by noting that $\log$ is a concave function to obtain the inequality. We then note that $|M_1|$ is a binomially distributed random variable with expectation $i \rho$.

Similar to what we did for $h(Y | Z_1, Z_2)$, we can also write
\[ h(Y | Z_2) = E_{M_2} \left[ h( X_\Sone^\top b_\Sone + X_{M_2}^\top b_{M_2} + W | M_2) \right] = E_{M_2} \left[ \frac{1}{2} \log\left(2 \pi e b_{\min} \left[ \sigma_x^2 i + \sigma_x^2 |M_2| + \frac{\sigma_w^2}{b_{\min}} \right] \right) \right]. \]

Combining the two entropy expressions, we then have
\begin{align*}
  I(Z_1 ; Y | Z_2) & \geq E_{M_2} \left[ \frac{1}{2} \log\left(2 \pi e b_{\min} \left[ \sigma_x^2 i + \sigma_x^2 |M_2| + \frac{\sigma_w^2}{b_{\min}} \right] \right) - \frac{1}{2}\log\left(2 \pi e b_{\min} \left[ \sigma_x^2 i \rho + \sigma_x^2 |M_2| + \frac{\sigma_w^2}{b_{\min}} \right] \right) \right] \\
  & = E_{M_2} \left[ \frac{1}{2} \log\left( \frac{i + |M_2| + \frac{\sigma_w^2}{\sigma_x^2 b_{\min}}}{i\rho + |M_2| + \frac{\sigma_w^2}{\sigma_x^2 b_{\min}}} \right) \right] = E_{M_2} \left[ \frac{1}{2} \log\left(1 + \frac{(1-\rho) i}{i \rho + |M_2| + \frac{\sigma_w^2}{\sigma_x^2 b_{\min}}} \right) \right] \\
  & \geq \frac{1}{2} \log\left(1 + \frac{(1-\rho) i}{i \rho + E_{M_2}[|M_2|] + \frac{\sigma_w^2}{\sigma_x^2 b_{\min}}} \right) = \frac{1}{2} \log\left(1 + \frac{(1-\rho) i}{i \rho + (K-i)\rho + \frac{\sigma_w^2}{\sigma_x^2 b_{\min}}} \right) \\
  & = \frac{1}{2} \log\left(1 + \frac{(1-\rho) i}{K\rho + \frac{\sigma_w^2}{\sigma_x^2 b_{\min}}} \right) = \frac{1}{2} \log\left(1 + (1-\rho) \frac{i b_{\min} \sigma_x^2}{\sigma_w^2 + \rho K b_{\min} \sigma_x^2} \right).
\end{align*}
Note that the expression above reduces to the expression for the fully observed case for $\rho = 0$.

Now consider the low SNR setup of \cite{wang,Rad-11}, where $\sigma_x = \sigma_w = 1$, $\bmin = 1/K$. Then, for the mutual information we have,
\[ I(Z_1 ; Y | Z_2) \geq \frac{1}{2} \log\left(1 + \frac{1-\rho}{1+\rho} \frac{i}{K} \right) = \Omega\left(\log\left(1 + \frac{1-\rho}{1+\rho} \right) \frac{i}{K} \right), \]
where the last equivalence can be shown by considering the two regimes $i/K = o(1)$ and $i/K = \Theta(1)$ separately. It then follows that $T = \Omega\left(\max_i \frac{i \log(N/i)}{\frac{i}{K} \log\left(1 + \frac{1-\rho}{1+\rho} \right)}\right) = \Omega\left(\frac{K \log N}{\log\left(1 + \frac{1-\rho}{1+\rho}\right)}\right)$ is sufficient for exact recovery for correlated $\bS$ and $K=O(N)$ or sufficient for recovery with a vanishing fraction of support errors for IID $\bS$ and $K = O(N/\log N)$, similar to the results in Section \ref{subsec:lcs} and \ref{subsec:correlated}.

\subsection{Proof of Theorem \ref{thm:qcs}}

In order to obtain the model-specific bounds, for a subset $|\St| = K-i$ we analyze the mutual information term $I_{\St,0}$ which is lower bounded by $I_\St(b)$ for any realization of $b \in \{-\sqrt{\bmin},\sqrt{\bmin}\}^K$.
Therefore, w.l.o.g.\ we consider $I_\St(b)$ for $b = \sqrt{\bmin} 1_K$ and omit the conditioning on $\bS = b$ and $S$ for brevity.

We write the mutual information term as
\begin{align*}
  I(X_\SSt ; Y | X_\St) = H(Y | X_\St) - H(Y | X_S)
\end{align*}
where we will analyze $H(Y| X_\St)$ and $H(Y | X_S)$ to obtain a lower bound for the mutual information expression.

Defining $Z_1 = \sqrt{\bmin} \sum_{j \in \SSt} X_j$, $Z_2 = \sqrt{\bmin} \sum_{j \in \St} X_j$ and $Z = Z_1 + Z_2$, we have $H(Y|X_\St) = H(Y|Z_2)$ since the quantizer input $X\beta + W$ depends only on the sum of the elements of $X_S$. Note that $Z_1 \sim{\cal N} (0,C_1^2)$ with $C_1^2 = \bmin i$, $Z_2 \sim{\cal N} (0,C_2^2)$ with $C_2^2 = \bmin (K-i)$. Now we explicitly write the conditional entropy
\begin{align}
  H(Y|Z_2) = \int_{-\infty}^{\infty} P_{Z_2}(z) H(Y|Z_2=z) \dx z = \int_{-\infty}^{\infty} P_{Z_2}(z) \left( p_1 \log \frac{1}{p_1} + p_0 \log \frac{1}{p_0} \right) \dx z
\end{align}
with $p_1 \triangleq \Pr[Y=1 | Z_2=z]$ and $p_0 \triangleq 1-p_1 = \Pr[Y=0 | Z_2=z]$, which can be written as
\begin{align*}
  p_1 = \Pr\left[ Z_1 + Z_2 + W \geq 0 \Big\vert Z_2=z \right] = \Pr\left[ Z_1 + W \geq -z \right] = \Pr\left[{\cal N} (0,E^2) \geq -z \right] = \mathcal{Q}\left( \frac{-z}{E} \right) \\
  p_0 = \Pr\left[  Z_1 + Z_2 + W < 0 \Big\vert Z_2=z \right] = \Pr\left[ Z_1 + W < -z \right] = \Pr\left[{\cal N} (0,E^2) < -z \right] = \mathcal{Q}\left( \frac{z}{E} \right)
\end{align*}
where $E^2 = \bmin i + 1$ and the $\mathcal{Q}$ function defined as $\mathcal{Q}(x) = \int_x^\infty \frac{1}{\sqrt{2\pi}} e^{-\frac{\tau^2}{2}} \dx \tau$. 

To lower bound $H(Y|Z)$, we make use of the following inequalities for $x > 0$ \cite{q1,q2}:
\begin{align}
  \frac{1}{12} \, e^{-x^2} \leq & \; \mathcal{Q}(x) \leq \frac{1}{2} e^{-\frac{x^2}{2}} \\
  \log 2 +\frac{x^2}{2} \leq \log(2e^{\frac{x^2}{2}}) \leq & \log \frac{1}{\mathcal{Q}(x)} \leq \log 12 + x^2.
\end{align}

Then, we write the following chain of inequalities:
\begin{align}
  H(Y|Z_2) & = 2 \int_{0}^{\infty} P_{Z_2}(z) \left( p_1 \log \frac{1}{p_1} + p_0 \log \frac{1}{p_0} \right) \dx z \label{eq:exQCS1}\\
  & \geq 2 \int_{0}^{\infty} P_{Z_2}(z) \cdot p_0 \log \frac{1}{p_0} \dx z \label{eq:exQCS2} \\
  & \geq 2 \int_{0}^{\infty} \frac{1}{\sqrt{2\pi C_2^2}} \cdot e^{-\frac{z^2}{2C_2^2}} \cdot \frac{1}{12} \cdot e^{-\frac{z^2}{E^2}} \cdot \left( \log 2 + \frac{z^2}{2E^2}\right) \dx z \label{eq:exQCS3} \\
  & = \frac{1}{12 \sqrt{2\pi} C_2} \int_{-\infty}^{\infty} e^{-A \frac{z^2}{2}} \cdot \left( \log 2 + \frac{z^2}{2E^2}\right) \dx z \label{eq:exQCS4} \\
  & = \frac{1}{12 \sqrt{2\pi} C_2} \left( \log 2 \frac{\sqrt{2\pi}}{\sqrt{A}} + \frac{\sqrt{\pi/2}}{A^{3/2} E^2} \right) = \frac{\log 2}{12 \sqrt{A} C_2} + \frac{1}{24 A^{3/2} C_2 E^2}\label{eq:exQCS5}
\end{align}
Equality \eqref{eq:exQCS1} follows from the evenness of the function inside the integral and we write \eqref{eq:exQCS2} by noting that $p_1 \log \frac{1}{p_1}$ and $P_{Z_2}(z)$ are non-negative. $P_{Z_2}(z)$ is expanded and the bounds above for the $\mathcal{Q}$ function are used to obtain \eqref{eq:exQCS3} and \eqref{eq:exQCS4} is a regrouping of terms by defining $A = \frac{1}{C_2^2} + \frac{2}{E^2}$ and rewriting the limits of the integral by noting that the integrand is an even function. We obtain \eqref{eq:exQCS5} by evaluating the integral. For $A$, we have
\begin{equation*}
  A = \frac{1}{\bmin (K-i)} + \frac{2}{\bmin i + 1} = s \frac{2K-i+s}{(i+s)(K-i)}
\end{equation*}
where we define $s \triangleq \frac{1}{\bmin}$ and replacing $A$, $C_2$ and $E$, we can then write
\begin{align}
  H(Y|X_\St) = H(Y|Z_2) & \geq \frac{\log 2}{12} \frac{\sqrt{i+s}}{\sqrt{2K-i+s}} + \frac{1}{24} \frac{\sqrt{i+s} (K-i)}{(2K-i+s)^{3/2}} \nonumber \\
  & = \frac{\sqrt{i+s}}{(2K-i+s)^{3/2}} \left( \frac{\log 2}{12} (2K-i+s) + \frac{1}{24} (K-i) \right) \\
  & \geq \frac{\log 2}{12} \frac{\sqrt{i+s}}{\sqrt{2K-i+s}}. \label{eq:exQCS_h1}  
\end{align}

We now analyze the second term $H(Y|X_S)$ to obtain an upper bound. Again, note that $H(Y|X_S) = H(Y|Z)$, then
\begin{equation*}
  H(Y|Z) = \int_\infty^\infty P_Z(z) H(Y|Z=z) \dx z = \int_\infty^\infty P_Z(z) \left(p_1 \log\frac{1}{p_1} + p_0 \log\frac{1}{p_0} \right)
\end{equation*}
where this time we define $p_1 \triangleq \Pr[Y=1|Z=z]$ and $p_0 \triangleq \Pr[Y=0|Z=z]$, which can be written as
\begin{equation*}
  p_1 = \Pr[Z+W \geq 0 | Z=z] = \Pr[W \geq -z] = \Pr[{\cal N}(0,1) \geq -z] = Q(-z)
\end{equation*}
\begin{equation*}
  p_0 = \Pr[Z+W < 0 | Z=z] = \Pr[W < -z] = \Pr[{\cal N}(0,1) < -z] = Q(z).
\end{equation*}

Then, write the following chain of inequalities:
\begin{align}
  H(Y|Z) & = 2 \int_0^\infty P_Z(z) \left( p_1 \log\frac{1}{p_1} + (1-p_1)\log\frac{1}{1-p_1}\right) \dx z \label{eq:exQCS6} \\
  & \leq 4 \int_0^\infty P_Z(z) \left( p_1 \log\frac{1}{p_1} \right) \dx z \label{eq:exQCS7} \\
  & \leq 4 \int_0^\infty \frac{1}{\sqrt{2\pi K}} e^{-\frac{z^2}{2K}} \frac{1}{2} e^{-\frac{z^2}{2}} \left( \log 12 + \frac{z^2}{2} \right) \dx z \label{eq:exQCS8} \\
  & = \frac{1}{\sqrt{2\pi K}} \int_{-\infty}^\infty e^{-B \frac{z^2}{2}} \left( \log 12 + \frac{z^2}{2} \right) \dx z \label{eq:exQCS9} \\
  & = \frac{\sqrt{s} \log 12}{\sqrt{2\pi K}} \frac{\sqrt{2\pi}}{\sqrt{B}} + \frac{\sqrt{s}}{2 \sqrt{2\pi K}} \frac{\sqrt{2\pi}}{B^{3/2}} = \frac{\sqrt{s} \log 12}{\sqrt{B K}} + \frac{\sqrt{s}}{2 \sqrt{K} B^{3/2}} \label{eq:exQCS10}
\end{align}
Equality \eqref{eq:exQCS6} follows from the evenness of the function inside the integral and we write \eqref{eq:exQCS7} by noting that $p \log \frac{1}{p} \geq (1-p) \log \frac{1}{1-p}$ for $0 \leq p \leq \frac{1}{2}$. $P_Z(z)$ is expanded and the above bounds for the $\mathcal{Q}$ function are used to obtain \eqref{eq:exQCS8} and \eqref{eq:exQCS9} is a regrouping of terms by defining $B = \frac{s}{K} + 1$ and rewriting the limits of the integral by noting that the integrand is an even function. We obtain \eqref{eq:exQCS5} by evaluating the integral. Replacing $B$, we then have
\begin{equation}
  H(Y|X_S) = H(Y|Z) \leq \log 12 \frac{\sqrt{s}}{\sqrt{K+s}} + \frac{1}{2} \frac{K \sqrt{s}}{(K+s)^{3/2}} \leq 2 \log 12 \frac{\sqrt{s}}{\sqrt{K+s}}. \label{eq:exQCS_h2}
\end{equation}

Looking at \eqref{eq:exQCS_h1} and \eqref{eq:exQCS_h2}, we have the following:
\begin{equation}
  I(X_\SSt ; Y | X_\St) \geq \frac{ \frac{\log 2}{12} \sqrt{i+s} - 2 \log 12 \sqrt{s} }{\sqrt{2K-i+s}},
\end{equation}
which is positive for all $i$ for a sufficiently large constant $\sqrt{\bmin} \approx 86$ with minimum occurring at $i = 1$. For large enough constant $\bmin$, we can then write $I(X_\SSt ; Y | X_\St) = \Omega(\sqrt{i/K})$.

Finally, since $\log\binom{N-K}{i}\binom{K}{i} = \Theta(i \log N)$, we have
\begin{equation*}
  \frac{\log\binom{N-K}{i}\binom{K}{i}}{I(X_\SSt ; Y | X_\St)} = O \left( \frac{i \log N}{\sqrt{i/K}} \right) = O(K \log N),
\end{equation*}
which is satisfied by $T = \Omega(K \log N)$, proving Theorem \ref{thm:qcs}.

\subsection{Proof of Theorem \ref{thm:group_testing}}

We start by considering the error exponent $E_N(\rho)$ for $(\Sone, \Stwo) \triangleq (\SSt, \St)$ for $|\SSt| = i \in \{1,\ldots,K\}$ and define $\alpha = i/K$ (not necessarily constant). We remark that $\bS$ does not exist for this problem, thus we can ignore the R\'enyi entropy term and the minimization over $b$.
\[ E_N(\rho) = -\log \left( \sum_Y \sum_{X_\Stwo} Q(X_\Stwo) \left[ \sum_{X_\Sone} Q(X_\Sone) p(Y | X_\Sone, X_\Stwo)^{\frac{1}{1+\rho}} \right]^{1+\rho} \right). \]
Notice that the observation model depends only on $Z_1 \triangleq \bigvee_{k \in \Sone} X_k$ and $Z_2 \triangleq \bigvee_{k \in \Stwo} X_k$ and $Z_1$ and $Z_2$ are Bernoulli random variables with parameters $p_1 = 1 - (1-1/K)^{\alpha K}$ and $p_2 = 1 - (1-1/K)^{(1-\alpha)K}$, respectively. Thus, we can rewrite the expression above as
\[ E_N(\rho) = -\log \left( \sum_Y \sum_{Z_2} P(Z_2) \left[ \sum_{Z_1} P(Z_1) p(Y | Z_1, Z_2)^{\frac{1}{1+\rho}} \right]^{1+\rho} \right), \]
where $p(Y | Z_1, Z_2) = 1$ if $Y = Z_1 \vee Z_2$ and 0 otherwise. Since this is a binary function, it is not affected by the exponentiation with $\frac{1}{1+\rho}$ and we can further simplify the expression above as
\[ E_N(\rho) = -\log \left( \sum_Y \sum_{Z_2} P(Z_2) \left[ \sum_{Z_1} P(Z_1) 1\{Y = Z_1 \vee Z_2\} \right]^{1+\rho} \right). \]

For the realization $y = 0$, the inner sum turns out to be $(1-p_2)(1-p_1)^{1+\rho}$, while for the realization $y=1$ we obtain $p_2 + (1-p_2) p_1^{1+\rho}$. We can then write the error exponent to be exactly $E_N(\rho) = -\log \left( p_2 + (1-p_2)\left[p_1^{1+\rho} + (1-p_1)^{1+\rho}\right] \right)$.
Notice that for large enough $K$, $p_1$ behaves as $1-e^{-\alpha}$, while $p_2$ behaves as $1-e^{-(1-\alpha)}$.

With this asymptotic consideration, and by letting $\rho = 1$, we obtain $e^{-(1-\alpha)}(1-e^{-\alpha})^2 + e^{-(1+\alpha)} + 1 - e^{-(1-\alpha)}$. With some algebra, we see that this is equal to $1 - \frac{2}{e} + \frac{2}{e}e^{-\alpha}$, thus we have $E_N(1) = \Theta(\alpha)$, for both $\alpha = \Theta(1)$ and $\alpha = o(1)$.

Now we can simply show that $T f_N(\rho) = T E_N(\rho) - \rho \log\binom{N-K}{i} - \log\binom{K}{i} - \log K \to \infty$ for $T = c K \log N$. For $\rho = 1$, it follows that $T f_N(1) = c~ i \log N - \Theta(i \log(N/i)) - \Theta(i \log N) - \log K \to \infty$ for a large enough constant $c$, thus proving the upper bound for $K = O(N)$. The lower bound can be obtained by noting that $I_\St \leq H(Y) \leq 1$ as for the binary regression case.

\bibliographystyle{IEEEtran}
\bibliography{references}

\begin{IEEEbiographynophoto}{Cem Aksoylar}
  received the B.Sc.\ degree in electronics engineering from Sabanci University, Istanbul, Turkey in 2010. He is currently working towards the Ph.D.\ degree in electrical and computer engineering in Boston University, MA. His research interests include statistical signal processing, high dimensional learning problems and inference and optimization over networks.
\end{IEEEbiographynophoto}

\begin{IEEEbiographynophoto}{George K.~Atia}
(S'01--M'09) received the B.Sc.\ and M.Sc.\ degrees from Alexandria University, Egypt, in 2000 and 2003, respectively, and the Ph.D.\ degree from Boston University, MA, in 2009, all in electrical and computer engineering.

He joined the University of Central Florida in Fall 2012, where he is currently an assistant professor and a Charles N.\ Millican Faculty Fellow in the Department of Electrical and Computer Engineering. From Fall 2009 to 2012, he was a postdoctoral research associate at the Coordinated Science Laboratory (CSL) at the University of Illinois at Urbana-Champaign (UIUC). His research interests include statistical signal processing, machine learning, stochastic control, wireless communications, detection and estimation theory, and information theory. 

Dr.\ Atia is the recipient of many awards, including the NSF CAREER Award in 2016, the Outstanding Graduate Teaching Fellow of the Year Award in 2003--2004 from the Electrical and Computer Engineering Department at Boston University, the 2006 College of Engineering Dean's Award at the BU Science and Engineering Research Symposium, and the best paper award at the International Conference on Distributed Computing in Sensor Systems (DCOSS) in 2008.
\end{IEEEbiographynophoto}

\begin{IEEEbiographynophoto}{Venkatesh Saligrama}
  is a faculty member in the Department of Electrical and Computer Engineering and Department of Computer Science (by courtesy) at Boston University. He holds a Ph.D.\ from MIT. His research interests are in statistical signal processing, machine learning and computer vision, information and decision theory. He has edited a book on Networked Sensing, Information and Control. He has served as an Associate Editor for IEEE Transactions on Information Theory, IEEE Transactions on Signal Processing and has been on Technical Program Committees of several IEEE conferences. He is the recipient of numerous awards including the Presidential Early Career Award (PECASE), ONR Young Investigator Award, and the NSF Career Award. More information about his work is available at http://sites.bu.edu/data.
\end{IEEEbiographynophoto}

\end{document}